%% file: arxiv_current.tex
\documentclass{article}

\input{arxiv_style}

\usepackage[utf8]{inputenc} %
\usepackage[T1]{fontenc}    %
\usepackage{url}            %
\usepackage{booktabs}       %
\usepackage{amsfonts}       %
\usepackage{nicefrac}       %
\usepackage{microtype}      %
\usepackage{xcolor}         %

\usepackage{listings}
\usepackage{amsmath,mathtools}
\usepackage{amsthm}
\usepackage{tikz}
\usepackage{caption,subcaption}
\usepackage{array}
\usepackage{mdwmath}
\usepackage{multirow}
\usepackage{mdwtab}
\usepackage{eqparbox}
\usepackage{multicol}
\usepackage{amsfonts}
\usepackage{multirow,bigstrut,threeparttable}
\usepackage{amsthm}
\usepackage{array}
\usepackage{bbm}
\usepackage{subcaption}
\usepackage{epstopdf}
\usepackage{mdwmath}
\usepackage{mdwtab}
\usepackage{eqparbox}
\usepackage{latexsym}
\usepackage{amssymb}
\usepackage{bm}
\usepackage{amssymb}
\usepackage{graphicx}
\usepackage{mathrsfs}
\usepackage{epsfig}
\usepackage{psfrag}
\usepackage{setspace}
\usepackage{tikz-cd}
\usepackage{url}

\usepackage{cleveref}
\usepackage{algorithm}
\usepackage{algorithmicx}
\usepackage{algpseudocode}
\usepackage[
autocite    = superscript,
backend     = bibtex, %
sortcites   = true,
style       = authoryear %
]{biblatex} %
\usepackage{nameref}
\usepackage[normalem]{ulem}

\usepackage{soul}

\input xy
\xyoption{all}

\theoremstyle{plain}

\def\EE{\mathbb{E}}

\def\PP{\mathbb{P}}

\def\RR{\mathbb{R}}

\def\bP{\mathbf{P}}

\def\bR{\mathbf{R}}

\def\bY{\mathbf{Y}}

\newcommand{\wt}[1]{\widetilde{#1}}

\renewcommand\l{\lambda}

\newcommand\diag{\textup{diag}}

\def\1{\mathbbm{1}}
\def\var{\mathsf{Var}}
\def\cov{\mathsf{Cov}}

\newcommand{\argmax}{\mathop{\mathrm{argmax}}}

\usepackage{booktabs}
\usepackage{adjustbox}
\usepackage{multirow}
\usepackage{listings}

\theoremstyle{plain}

\def \cF {\mathcal{F}}
\def \bP {\mathbb{P}}
\def \bR {\mathbb{R}}

\def \cN {\mathcal{N}}
\def \var {\mathsf{Var}}

\usepackage{xspace}

\newcommand{\TV}{{\sf TV}}

\definecolor{myblue}{rgb}{.8, .8, 1}
\definecolor{mathblue}{rgb}{0.2472, 0.24, 0.6} %
\definecolor{mathred}{rgb}{0.6, 0.24, 0.442893}
\definecolor{mathyellow}{rgb}{0.6, 0.547014, 0.24}

\usepackage{enumitem}

\usepackage{pgfplots}
\pgfplotsset{compat=1.17}

\input{jian}
\usepackage{color-edits}
\addauthor{jq}{brown}
\addauthor{gs}{green!60!black}
\let\oldparagraph\paragraph
\renewcommand{\paragraph}[1]{\oldparagraph{#1.}}

\title{Bridging multiple worlds: multi-marginal optimal transport for causal partial-identification problem}

\author{%
    Zijun Gao\\
    USC Marshall\\
    {\small\texttt{zijungao@marshall.usc.edu}}
    \and
    Shu Ge\\
    MIT EECS\\
    {\small\texttt{geshu@mit.edu}}
    \and
    Jian Qian\\
    MIT EECS\\
    {\small\texttt{jianqian@mit.edu}}
}
\date{}

\begin{document}

\maketitle

\begin{abstract}
Under the prevalent potential outcome model in causal inference, each unit is associated with multiple potential outcomes but at most one of which is observed, leading to many causal quantities being only partially identified.
The inherent missing data issue echoes the multi-marginal optimal transport (MOT) problem, where marginal distributions are known, but how the marginals couple to form the joint distribution is unavailable.
In this paper, we cast the causal partial identification problem in the framework of MOT with $K$ margins and $d$-dimensional outcomes and obtain the exact partial identified set. 
In order to estimate the partial identified set via MOT, statistically,
we establish a convergence rate of the plug-in MOT estimator for the $\ell_2$ cost function stemming from the variance minimization problem and prove it is minimax optimal for arbitrary $K$ and $d \le 4$. We also extend the convergence result to general quadratic objective functions.
Numerically, we demonstrate the efficacy of our method over synthetic datasets and several real-world datasets where our proposal consistently outperforms the baseline by a significant margin (over 70\%).
In addition, we provide efficient off-the-shelf implementations of MOT with general objective functions. \looseness=-1
\end{abstract}

\section{Introduction}\label{sec:introduction}

The potential outcome model \parencite{rubin1974estimating} has been extensively used to perform causal inference.
Suppose there are $K$ possible treatment levels,
and each unit $i$ is associated with $K$ potential outcomes $\bY_i(\cdot) = (Y_i(1), \ldots, Y_i(K))$, where each potential outcome $Y_i(k) = (Y_i^1(k), \ldots, Y_i^d(k))$ is a $d$ dimensional vector drawn independently\footnote{In addition to the super-population scenario, we discuss finite sample analysis in \Cref{eg:design-based-inference}.} from an unknown joint distribution $\mustar$.
Causal questions focus on comparing these potential outcomes to evaluate the treatment effect.

The fundamental challenge of causal inference is that we can only observe the potential outcome $Y_i = Y_i(W_i)$ at the realized treatment level $W_i \in \{1,\ldots, K\}$, and it is impossible to observe the entire set of $K$ potential outcomes.
As a consequence, only the marginal distribution $\mustar(k)$ of $Y_i(k)$, $1 \le k \le K$ can be identified and how the marginals $\mustar(k)$ couple to form the full joint distribution $\mustar$ is not identifiable.
This renders a large class of causal estimands\footnote{
Only causal estimands that can be expressed as the expectation of the sum of univariate functions of a single potential outcome can be identified.} relying on the joint distribution $\mustar$ to be only partially identifiable, meaning that the causal estimands can not be precisely determined but can only be confined to a range of plausible values based on the available data.
For partially-identified causal estimands, the statistical task is determining the identified set: the full set of values that are compatible with the observed data.

In the following, we provide an example of a causal quantity that could be only partially identified and whose partially identified set hasn't been fully characterized in the literature. 

\subsection{Motivating example}\label{eg:design-based-inference}

To illustrate the concept of partial identifiability, we revisit the classic Neymanian confidence interval \parencite{neyman1923application}.
Consider $d = 1$ and the average contrast effect estimand ${\tau}_{\boldsymbol{\beta}} :=  
\sum_{k=1}^K \beta_k \bar{Y}(k)$, with $\bar{Y}(k) = \sum_{i=1}^n {Y}_i(k)/n$ and pre-fixed coefficients $\beta_k \in \RR$. Suppose $W_i$ are drawn from a completely randomized experiment\footnote{A completely randomized experiment involves a random sample of size $n_k$ chosen without replacement from a finite sample of size $n$.} with $n_k$ units at the treatment level $k$, then
the Neymanian confidence interval for ${\tau}_{\boldsymbol{\beta}}$ takes the form \parencite{lu2016randomization},
\begin{align}\label{eq:variance.design.based}
\begin{split}        
    V_{\boldsymbol{\beta}} :=\sum_{k=1}^K  \frac{\beta_k^2 S_k^2}{(n/K)} - \frac{S_{\tau}^2}{n}, 
    \quad S_k^2 = \frac{\sum_{i=1}^n (Y_i(k) - \bar{Y}(k))^2}{n-1}, 
    S_{\tau}^2 = \frac{\sum_{i=1}^n (\sum_{k=1}^K \beta_k Y_{i}(k) - \tau_{\boldsymbol{\beta}})^2}{n-1}.  
\end{split}
\end{align}
The variances $ S_k^2 $ can be estimated using the empirical variance of units with $ W_i = k $. However, the variance $ S_{\tau}^2 $, remains unidentifiable because the outcomes $ Y_i(k) $ for different $ k $ are not simultaneously observable for the same individual. The conventional estimator of the variance $ V_{\boldsymbol{\beta}} $ substitutes $ S_k^2 $ with the empirical variance and replaces $ S_{\tau}^2 $ by a trivial lower bound zero. The conventional variance estimator is upward-biased and thus induces an over-conservative confidence interval.

For cases where $ K = 2 $ (one treatment level and one control level), \cite{aronow2014sharp} provide the sharp lower bound for $ S_{\tau}^2 $. 
Our proposed method below can be used to generalize the result to arbitrary $K$ levels to offer a sharp lower bound of $ S_{\tau}^2 $ as well as the narrowest confidence intervals. 
Practically, a narrower confidence interval often indicates that a smaller sample size, such as less patients to involve in a clinical trial, could be sufficient to make a discovery with statistical significance, thereby potentially reducing the associated costs.

\subsection{Our proposal}

This paper addresses the partial identification problem under the potential outcome model using Multi-margin Optimal Transport (MOT).
Optimal transport, particularly MOT, is a fast-growing research area partly due to the emerging applications in machine learning including networks (GANs) \parencite{choi2018stargan}, domain adaptation \parencite{hui2018unsupervised}, and Wasserstein barycenters \parencite{agueh2011barycenters}.
The MOT problem considers a cost (objective) function of random variables with known marginal distributions and seeks to find the coupling of the marginal distributions that minimizes the expected cost \parencite{pass2015multi}. 
In addition, MOT outputs the exact lower bound of the expected cost among all joint distributions obeying the marginal distributions, even if the optimal coupling is not unique.

There is a natural connection between the partially identified set and MOT: the marginal distribution of each potential outcome is identifiable but the joint distribution can not be uniquely determined.
For a causal estimand, there is an MOT problem whose optimal objective value corresponds to the lower (upper) limit of the causal quantity that can be achieved by some joint distribution respecting the identifiable marginals.
Particularly, let $\Gamma$ represent the collection of joint probabilities respecting the marginal distributions $\mustar(k)$, $1 \le k \le K$. 
For a causal estimand of the form $\mathbb{E}_{\mustar}[\ell(\bm Y_i(\cdot))]$, where $\ell$ can depend on all potential outcomes, the identified set can be expressed as
\begin{align*}
    \left\{\mathbb{E}_\gamma[\ell(\bm Y_i(\cdot))]: \gamma \in \Gamma \right\}.
\end{align*}
Since $\Gamma$ is convex by definition and the estimand $\mathbb{E}_\gamma[\ell(\bm Y_i(\cdot))]$ is linear in $\gamma$, the identified set is always an interval\footnote{We allow the interval to include $-\infty$ and $\infty$. We also allow the interval to be degenerate (only containing a single point), or even empty.}
\begin{align}\label{eq:interval}
    \left[\inf_{\gamma \in \Gamma}\mathbb{E}_\gamma[\ell(\bm Y_i(\cdot))],
    ~\sup_{\gamma \in \Gamma} \mathbb{E}_\gamma[\ell(\bm Y_i(\cdot))]\right].
\end{align}
The task of determining the identified set
reduces to calculating the lower and upper limits, which can be obtained by solving two MOT problems.
See \Cref{fig:diagram} for visualization.

\begin{figure}[h]
        \centering
    \begin{minipage}{0.96\textwidth}
                \centering
                \includegraphics[clip, trim = 0cm 1cm 0cm 1cm, width = \textwidth]{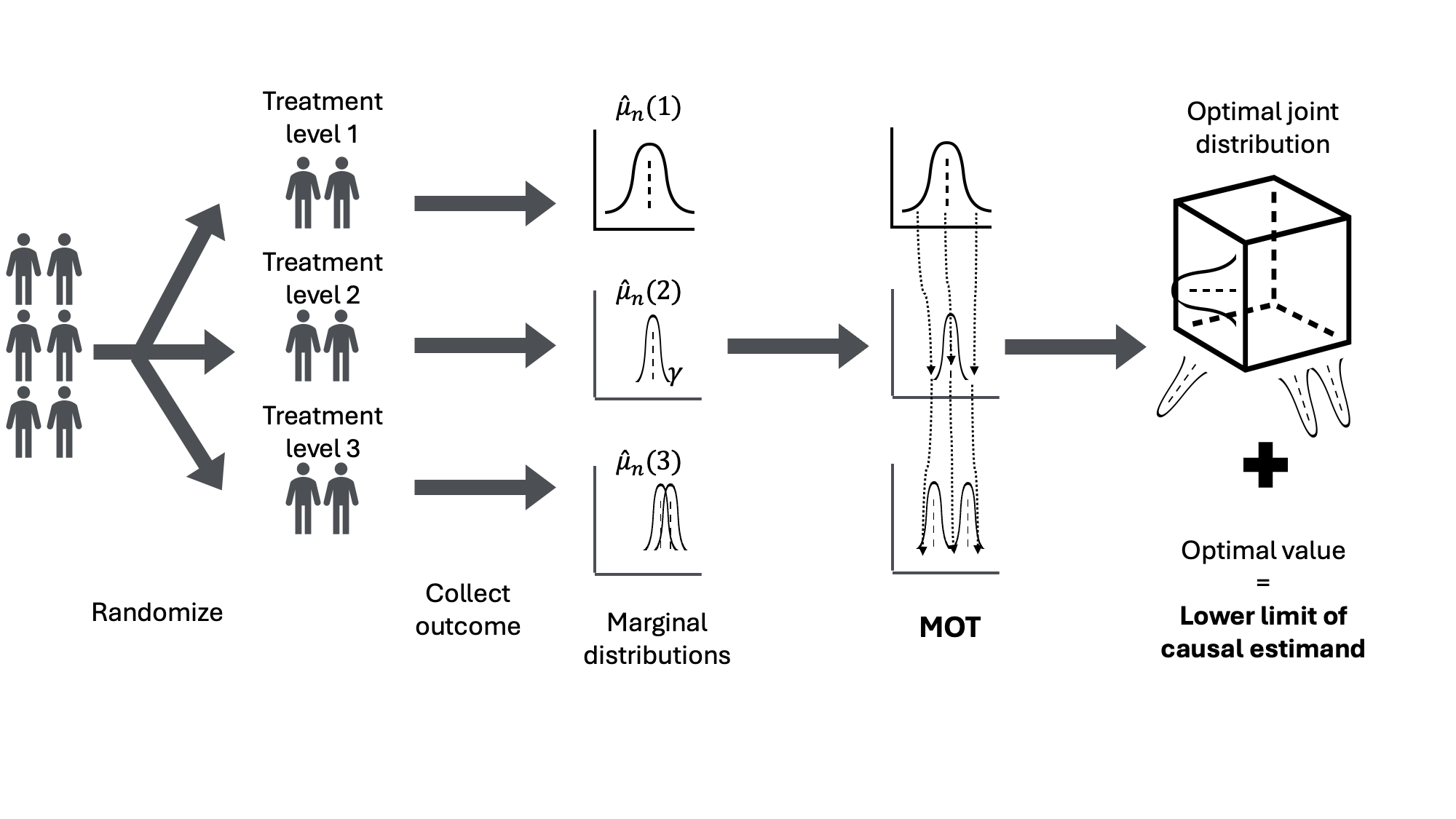}
        \end{minipage}
        \caption{Diagram of solving partial identification problem via MOT.
        A treatment with multiple levels is applied to a group of individuals. The outcomes of each treatment group are collected to form the empirical marginal distribution denoted by $\hat{\mu}_n(k)$, $1 \le k \le K$.
        An MOT problem with the causal estimand as the objective function and empirical marginal distributions as the constraints are solved to obtain the lower limit of the partially identified set. 
        The upper limit can be obtained similarly by flipping the sign of the MOT objective function.
        }
    \label{fig:diagram}
\end{figure}

We illustrate the connection between MOT and the partial identification problem using the motivating example in \Cref{eg:design-based-inference}. For the partially identified $ S_{\tau}^2 $, the loss function associated is $ \ell(\bm{Y}_i(\cdot)) := (n/(n-1)) (\sum_{k=1}^K Y_i(k))^2 - \tau_{\beta}^2 $. Here $ \tau_{\beta} $ is identifiable, and 
the family of joint distributions $ \Gamma $ includes all the distributions whose $k$-th marginal distribution equals the distribution of $ Y_i(k)$. 
Determining the lower bound of $ S_{\tau}^2 $ is thus equivalent to minimizing $ \EE_{\gamma}[\ell(\bm{Y}_i(\cdot))]$ for $\gamma \in \Gamma$, which becomes a standard MOT problem.
Theoretically, for this specific case, we show that the plug-in estimator (which computes the MOT value of the empirical marginal distributions) for the MOT optimal objective value converges at a rate of $O(n^{-1/2})$. 
Empirically, we apply our proposal to the Expitaxial layer growth dataset \parencite{wu2011experiments} in \Cref{sec:empirical}.
Our variance estimator using the sharp lower bound of $S_\tau^2$ sees a reduction of $8.16$ percent compared to the conventional variance estimator (lower bounding $S_\tau^2$ by zero), effectively decreasing the number of units or repeats required by the same percentage.

\subsection{Contributions}\label{sec:contribution}
\begin{enumerate}

    \item Methodologically, our proposal recovers the exact partially identified set given the true marginal distributions. 
    Our proposal can accommodate various scenarios including arbitrary objectives (beyond quadratic), multiple treatment levels, multiple treatments, and multiple outcomes (see \Cref{sec:setup} for examples).
    Concretely, we demonstrate that our proposal can be used to detect contrast treatment effect and interaction effect, analyze the covariance between different treatment effects, and shorten Neymanian confidence intervals (\Cref{eq:variance.design.based}).

    \item Theoretically, we analyze the plug-in estimator associated with the MOT problem where the true marginal distributions are replaced by the empirical marginal distributions.
    In particular, we focus on the $\ell_2$ norm of the average potential outcomes as the objective. We show the plug-in estimator converges at the parametric rate $O(n^{-1/2})$ up to logarithmic terms for $d\leq 4$, which is surprisingly independent of the number of margins $K$. This is complemented with a matching lower bound, demonstrating that the plug-in estimator is minimax optimal for the particular case. 
    The proof techniques involved,  as elaborated in \Cref{sec:theory}, could be of independent interest. 
    The convergence of general quadratic objectives is also provided as an extension of \cite{chizat2020faster}.

    \item Implementation-wise, we provide several solvers\footnote{The implementation can be accessed at \url{https://github.com/ShuGe-MIT/mot_project}.
    } of MOT applicable to arbitrary objectives (not limited to quadratic objectives), addressing the lack of efficient off-the-shelf MOT packages.  
     A notable feature is that our algorithm 
     always provides valid, albeit conservative, partially identified sets. 
    On synthetic datasets, we compare our method to a case where the closed-form solution of the optimal values is known, showing that our estimates closely match the true optimal values, and in a scenario where the exact solution is unknown, we demonstrate that our method's bounds are valid for the true value of the causal estimand (see details in \Cref{sec:empirical}).
    On real datasets, we validate our proposal on the aforementioned applications
    and showcase that our method achieves an improvement of approximately $70\%$ or more compared to a baseline approach\footnote{
    The causal estimands in our applications under no model assumptions have not been thoroughly investigated in the literature and the baseline bound used is the most relevant counterpart to compare with to the authors' knowledge.} detailed in \Cref{sec:setup}.

\end{enumerate}

\paragraph{Notations} We summarize the notations in \cref{app:notation}.

\subsection{Related works}

\subsubsection{Partial-identification}
In the literature of causal inference, identifiability is typically the first issue addressed when a new causal estimand is proposed. 
Under the potential outcome model, causal estimands depending on cross-world potential outcomes are generally not identifiable due to the inherent missingness.
In this case, the range of parameter values that are compatible with the data is of interest.
A multitude of works have been dedicated to finding causal estimands that can be identified, where a set of assumptions is usually necessary.
For instance, the identifiability of the average treatment effect requires the assumption of no unobserved confounders, and using instrumental variables to identify the local treatment effect requires the exogeneity \parencite{imbens2015causal}.
However, verifying these assumptions can be challenging or even impossible based solely on the data, and the violation of the assumptions will render the parameter unidentified.

In the literature of econometrics, the partial identification problem has been a topic of heated discussion for decades (see \cite{kline2023recent} for a comprehensive review and references therein). 
The works generally begin with a model depending on the parameter of interest and possibly other nuisance parameters or functions.
The partially identified sets are specified to be the values that comply with a set of moment inequalities, parameters that maximize a target function, or intervals with specifically defined endpoints. 
In causal inference, there is a tendency to avoid heavy model assumptions imposed on the potential outcomes to ensure the robustness and generalizability of causal conclusions.
Methods therein often involve leveraging non-parametric or semi-parametric methods that make few structural assumptions about the data-generating process \parencite{bickel1993efficient, van2000asymptotic, laan2003unified}.

We detail two particularly relevant threads of works among the studies on partial identification problems.
Copula models \parencite{nelsen2006introduction} are used to describe the dependency between multiple random variables with known marginal distributions, which also aligns with the potential outcome model with marginals accessible and the coupling unknown.
The Fréchet–Hoeffding copula bounds can be used to characterize the joint distribution of the potential outcomes \parencite{heckman1997making, manski1997mixing, fan2010sharp}. However, when dealing with more than two margins, one side of the Fréchet–Hoeffding theorem's bound is only point-wise sharp; moreover, copula models are generally constrained to unidimensional random variables; in addition, the bounds on the joint distribution do not necessarily translate to those of causal estimands, such as the variance of the difference of two potential outcomes. 
The second thread of related works explicitly employs the optimal transport
to address the partial identification problem in econometrics and causal inference. 
\cite{galichon2018optimal}  considers a model that includes observed variables and latent variables, and the parameter of interest regulates the distribution of the latent variables and the relationship between the observed and unobserved variables, which differ from our problem formulation.
\cite{balakrishnan2023conservative} introduce the quadratic and differential effects for continuous treatments (infinite levels of treatment) that can be considered as examples of the lower limits of our formulation \eqref{eq:interval} with a quadratic objective function.

    \subsubsection{MOT}
    \paragraph{MOT and minimax estimation} 
    The convergence of the empirical optimal transport cost (with two margins) to its population value has been extensively studied and well-understood (see \cite{chizat2020faster,staudt2023convergence,manole2024sharp} and the references therein). However, the estimation error of the plug-in estimator, which is the empirical optimal transport cost, for general loss functions remains a long-standing open problem that is only partially understood under strong assumptions \parencite{staudt2023convergence,manole2024sharp}. Notably, it is known to be minimax optimal for the Wasserstein $p$-cost.
    However, to the best of our knowledge, such minimax results have not been established for the MOT problem. 
    In this paper, we generalize the approach by \cite{chizat2020faster} to provide a convergence result for the plug-in estimator of the MOT problem for arbitrary quadratic objective functions (Eq.~\eqref{eq:quadratic.type}). Furthermore, we refine the analysis for the average squared $\ell_2$ norm (Eq.~\eqref{eq:quadratic.objective}) and show that the plug-in estimator is minimax optimal when the dimension $d\leq 4$. Determining the minimax rates for cases where $d>4$, and for arbitrary quadratic functions or arbitrary cost functions remains an interesting direction for future research.

    \paragraph{Computation for MOT} 
    For computing Optimal Transport (OT) distance, the Sinkhorn algorithm \parencite{cuturi2013sinkhorn} and its accelerated variants (see \cite{altschuler2017near,lin2019efficient,xie2020fast} and references therein) remain the state-of-the-art approach. Extensions of the Sinkhorn algorithm to the MOT problem have been proposed in the literature \parencite{peyre2019computational,tupitsa2020multimarginal,lin2022complexity}. Notably, \cite{tupitsa2020multimarginal,lin2022complexity} provide finite time convergence guarantees for the Multi-marginal Sinkhorn and Primal-Dual Accelerated Alternating Minimization (PD-AAM) algorithms, respectively. While the implementations of OT algorithms are quite advanced \parencite{flamary2021pot}, those for MOT algorithms are still in the early stages. Specifically, \cite{tupitsa2020multimarginal} offers implementations that work with $K=4$, $n=60$ and accuracy $\eps = 1/80$. Since our tasks require concrete lower bounds for the MOT problem at a considerably larger scale of sample size $n$ and higher accuracy $\eps$, we have implemented several algorithms (see \Cref{sec:implementation} for details), used them to obtain our numerical results, and provided an efficient off-the-shelf MOT package.

\paragraph{Organization} The paper is organized as follows. 
In \Cref{sec:setup}, we fix notations and frame various causal partially identified set problems as a MOT problem.
In \Cref{sec:theory}, we investigate the estimation accuracy of the plug-in estimator of the MOT problem. 
In \Cref{sec:empirical}, we discuss implementations and application results on real world datasets.
In \Cref{sec:discussion}, we conclude the paper by discussing potential research directions stemming from the current work.

\section{MOT formulation of partial identification problem}\label{sec:setup}

We recall briefly the relevant notations. We denote by $\bm Y_i(\cdot)$ the potential outcomes of unit $i$. At treatment level $k$, the potential outcome is denoted by $Y_i(k) = (Y_i^1(k),...,Y_i^d(k)) \in \RR^d$. Let $\mustar$ be the joint distribution that are only partially identifiable and $\mustar(k)$ be its $k$-th margin for $k\in [K]$. 
Without loss of generality, we suppose the marginals $\mustar(k)$ for any $k\in [K]$ are supported on $B_1^d(0)\subset \bR^d$ which is the $d$-dimensional unit ball.

We explore three types of objectives in the increasing order of specificity (\Cref{tab:summary}).

\begin{table}[htp]
\centering    
\caption{Summary of results for different objective functions.
}
\label{tab:summary}
\begin{tabular}{l|p{5cm}p{3.5cm}}
\toprule
Objective function & Theoretical properties  & Computation \\ \midrule
Arbitrary quadratic function & Upper bound (\Cref{prop:quadratic-form-upper-bound-simplified})  &   
\multirow{2}{3.5cm}{\Cref{alg:main-alg} with provable convergence}     
\\
Average squared $\ell_2$ norm & Upper bound (\Cref{thm:main-upper-bound}), lower bound (\Cref{prop:minimax-lower-bound})  &  \\ \bottomrule
\end{tabular}
\end{table}

\subsection{Arbitrary function}\label{sec:arbitrary.function}
In \Cref{sec:introduction}, we described the causal partial identification problem where the parameter of interest is the expectation of an arbitrary objective function $\ell$. 
For the two-margin optimal transport with general objective function, the convergence of the empirical optimal transport cost has been established for smooth objective functions with compact supports \parencite{hundrieser2022empirical}, objectives of the form $\ell(y(2) - y(1))$ \parencite{manole2024sharp}, or under bounded moment assumptions \parencite{staudt2023convergence}.
To the authors' knowledge, the convergence results for general MOT cost has not been explicitly stated.
For the rest of the paper, we will focus on quadratic objective functions that possess favorable statistical properties and admit various applications.

\subsection{Arbitrary quadratic function}

Quadratic objective functions are practically useful (as demonstrated in the examples below) and the estimation error analysis can be helped by the quadratic structure.
Let $A \in \RR^{K \times K}$ be a matrix and recall that $\Gamma(\mustar(1),...,\mustar(K))$ denotes the set of all possible couplings of $\mustar(1),...,\mustar(K)$, i.e., distributions $\gamma$ where the $k$-th margin $\gamma(k)$ is equal to $\mustar(k)$ for all $k\in [K]$.
We define the Quadratic Multi-marginal Wasserstein ($\QMW$) loss parametrized by $A$ as
\begin{align}\label{eq:quadratic.type}
    \QMW(\mustar; A) = \QMW(\mustar(1),\dots,\mustar(K); A) := \inf_{\gamma \in \Gamma(\mustar(1),\dots,\mustar(K))}  
    \int y^\top A y ~\gamma(dy),
\end{align}
where $y^\top A y$ is defined as $y^\top A y\ldef \sum_{1\leq i,j\leq K} a_{ij} y(i)^\top y(j) $.

\subsection{Average squared $\ell_2$ norm}

Specifically, we focus on the following quadratic form which is flexible enough to encompass most of our examples of interest (\Cref{eg:design-based-inference,eg:interaction-effect}) as well as objectives in existing literature detailed below. Concretely, we define Multi-marginal Wasserstein ($\MW$) loss as
\begin{align}\label{eq:quadratic.objective}
    \MW_2^2(\mustar) = \MW_2^2(\mustar(1), \ldots, \mustar(K)) := \inf_{\gamma \in \Gamma(\mustar(1),\dots,\mustar(K))} \int \left\| \frac{\sum_{k=1}^{K} y(k)}{K} \right\|_2^2  ~\gamma(dy),
\end{align}
where the objective function in Eq.~\eqref{eq:quadratic.objective} corresponds to the integral of the squared $\ell_2$ norm of the average potential outcome $(\sum_{k=1}^{K} y(k))/K$.
This objective is a special case of in \eqref{eq:quadratic.type} with $A = (a_{ij}=1/K^2)_{1\leq i,j\leq K}$.

The quantity $\MW_2^2(\mustar)$ has multiple connections to the existing literature. For instance, the expected sum of pairwise squared Wasserstein distance $\sum_{k \neq k'} \EE_{\mustar}[(Y_i(k) - Y_i(k'))^2]$ is related to the repulsive harmonic cost arising from the weak interaction regime in Quantum Mechanics~\parencite{di2015optimal} and can be written as $K^2 \MW_2^2(\mustar)$ subtracted from the estimable quantity $2 K \sum_{k=1}^K \EE_{\mustar(k)} [(Y_i(k))^2]$. 
Another example is variance minimization problem~\parencite{ruschendorf2013mathematical}, where the minimum variance can be written as $K^2 \MW_2^2(\mustar)$ subtracting the estimable quantity $\prn{\sum_{k=1}^K \EE_{\mustar(k)} \brk{Y_i(k)} }^2$. Lastly, for the heritability detailed in \Cref{eg:heritability}, 
the upper bound of the metric can be obtained by solving a modified version of
Eq.~\eqref{eq:quadratic.objective} incorporating the weights $\PP(W_i = k)$.

We describe a non-tight lower bound of the objective Eq.~\eqref{eq:quadratic.objective}, which serves as the baseline to compare with.
The objective Eq.~\eqref{eq:quadratic.objective} admits the following decomposition (\Cref{lem:mean-decomposition}),
\begin{align}\label{eq:decomposition}
\MW_2^2(\mustar) = \left\| \frac{\sum_{k=1}^{K} \EE_{\mustar(k)} [Y_i(k)]}{K}  \right\|_2^2 + \MW_2^2(\bar{\mu}_\star),
\end{align}
where $\bar{\mu}_\star$ denotes the centered counterpart of ${\mu}_\star$, that is, $\bar{\mu}_\star(k)$ is a translation of $\mustar(k)$ by $-\En_{\mustar(k)}\brk*{Y_i(k)}$. 
The first term on the right-hand side of Eq.~\eqref{eq:decomposition} depends solely on the first moment of each margin, and the second term contains the information of higher-order moments.
Since $\MW_2^2(\bar{\mu}_\star)$ is non-negative, the first term is a straightforward lower bound of $\MW_2^2(\mustar)$, which we refer to as the baseline lower bound.
When $\MW_2^2(\bar{\mu}_\star)$ is zero, i.e., $\mustar(k)$ satisfying the joint mixability property \parencite{wang2013bounds}, the baseline lower bound is adequate.
When the marginal distributions differ in second or higher-order moments, the baseline is not sufficient and our method offers a potentially significantly tighter lower bound (see \Cref{sec:partial.identification} for real data demonstrations).

\subsection{Examples}
We provide examples of causal partial-identification problems that can be formulated in terms of the objective in 
Eq.~\eqref{eq:quadratic.objective} or \eqref{eq:quadratic.type}. 
\looseness=-1

\begin{example}[Detection of interaction treatment effect]
\label{eg:interaction-effect}
Suppose there are two binary treatments.
We use $Y_i(0,0)$ to denote the potential outcome under control of both treatments, and similarly for $Y_i(0,1)$, $Y_i(1,0)$, $Y_i(1,1)$, and the potential outcomes $Y_i(\cdot) = (Y_i(1,1), -Y_i(1,0), - Y_i(0,1), Y_i(0,0))$ 
possess four margins in total.
Identifying a non-zero interaction treatment effect, i.e., $Y_i(1,1) - Y_i(1,0) - Y_i(0,1) + Y_i(0,0) \neq 0$, is crucial as it can facilitate the detection of drug synergy (whether different drugs' combined effect is greater than the sum of their individual effects).
The tight lower bound of $\EE_{\mustar}[(Y_i(1,1) - Y_i(1,0) - Y_i(0,1) + Y_i(0,0))^2]$ can be solved for as in Eq.~\eqref{eq:quadratic.objective}, and a significantly positive lower bound suggests the existence of interaction effect.
\end{example}

\begin{example}[Detection of contrast effect]\label{eg:contrast-effect}
The method in \Cref{eg:interaction-effect} can also be applied to detect the contrast effect of a treatment with multiple levels.
Suppose there is a treatment with $K$ levels, and the corresponding potential outcomes are associated with $K$ margins. 
Unlike the case of a binary treatment, there are various ways to define an estimand for treatments of multiple levels, and a popular choice is considering a contrast vector $\boldsymbol{\beta} \in \RR^k$, $\sum_{k=1}^K \beta_k = 0$ and the associated contrast treatment effect $
\sum_{k=1}^K \beta_k {Y}_i(k)$.
Similar to the interaction effect, a significantly positive lower bound of $\EE_{\mustar}[(\sum_{k=1}^K \beta_k {Y}_i(k))^2]$ is a sign to reject the null hypothesis of no contrast effect.

We remark that standard methods for testing the existence of treatment effects construct a confidence interval for the treatment effect estimator $\sum_{i=1}^n (n/K)^{-1} \sum_{k=1}^K \beta_k \1_{\{W_i = k\}} Y_i$ \parencite{imbens2015causal}. 
This is testing whether the average contrast treatment effect $\EE_{\mustar}[\sum_{k=1}^K \beta_k {Y}_i(k)]$ is zero but we are testing whether all the individual contrast treatment effect is zero, i.e., $\sum_{k=1}^K \beta_k {Y}_i(k) = 0$ for all $i$. 
\end{example}

\begin{example}[Analysis of variance (ANOVA)]
Analysis of variance (ANOVA) is a standard topic in statistics to determine the importance of different variables, and recent research \parencite{hines2022variable} has focused on using ANOVA to define the importance of various treatments. Particularly, \cite{hines2022variable} use variance to quantify the importance, which is a special case of the quadratic form (with the identifiable mean squared subtracted).
Our work can be used to provide lower and upper bounds of the importance of a treatment.    
\end{example}

\begin{example}[Heritability]\label{eg:heritability}
Heritability is a statistical measure frequently used in genetics to estimate the proportion of variation in a trait among individuals in a population that is due to genetic differences. In this application, the outcome is a biological trait, treatments are genes and environmental factors, and the heritability corresponds to a quadratic objective of our formulation.
One metric to quantify the heritability in genetic studies \parencite{jacquard1983heritability} is defined as $\EE_{\mustar}[\var(Y_i(W_i) \mid \bm Y_i(\cdot))]$, or equivalently $\EE_{\mustar}[Y_i^2(W_i)] - \EE_{\mustar}[(\sum_{k=1}^K Y_i(k) \PP(W_i = k))^2]$. 
This corresponds to the MOT problem with a quadratic objective function.
\end{example}

\begin{example}[Covariance between treatment effects]
\label{eg:cov-between-te}
    The covariance between treatment effects of different dimensions of the outcomes $\cov_{\mustar}(\tau^1_i, \tau^2_i)$ with $\tau_{i}^j = Y_{i}^j(2) - Y_{i}^j(1)$, $j = 1$, $2$ takes the form of Eq.~\eqref{eq:quadratic.type} and its upper and lower bounds can be obtained by solving the corresponding MOT problems.

    By investigating the exact lower and upper bounds of the covariance, we provide a model-free way to detect positively or negatively associated treatment effects (\Cref{sec:covariance}).
    There are many relevant applications in economics, e.g., the effects of a promotion scheme on the sale of complementary goods shall be positively correlated, and social sciences, e.g., the effects of a subsidy program on the full-time employment and part-time employment rate, might be negatively correlated.
\end{example}

\section{Theoretical characterization of the estimation error}\label{sec:theory}

In this section, we consider the minimax estimation error of the value  $\MW_2^2(\mustar)$ and $\QMW(\mustar;A)$ with $n$ empirical samples. Concretely, suppose we have $n$ i.i.d. samples from the distribution $\mustar$ and assigned $n/K$ samples to each treatment level $k$. For simplicity, we assume $n/K$ is an integer. Specifically, suppose we have samples $\set{Y_i(k)}_{i\in [n/K],k\in [K]}$. Let the empirical distribution be $\muhat_n  = (\muhat_n(1),..., \muhat_n(K))$, where $\muhat_n(k) = \unif(\set{Y_i(k)}_{i\in [n/K]})$. In \Cref{sec:upper-bound}, we establish the parametric rate of convergence for $\MW_2^2(\muhat_n)$ to $\MW_2^2(\mustar)$ for $d\leq 4$ and convergence result for $\QMW(\muhat_n;A)$ to $\QMW(\mustar;A)$, respectively. Then we present a matching minimax lower bound for the estimation error of $\MW_2^2(\mustar)$ when $d\leq 4$ in \Cref{sec:lower-bound}.

\subsection{Upper bound}
\label{sec:upper-bound}

In this section, we prove finite sample upper bounds for the convergence of $\MW_2^2(\muhat_n)$ and $\QMW(\muhat_n;A)$ to $\MW_2^2(\mustar)$ and $\QMW(\mustar;A)$, respectively. Similar convergence results have been shown for optimal transport \parencite{chizat2020faster}. 
For the convergence of $\MW_2^2(\muhat_n)$, we first apply the techniques from \cite{chizat2020faster} to obtain a suboptimal (sometimes vacuous) convergence rate and then refine their analysis using more detailed properties of the MOT. Concretely, for $d=1$, applying techniques from \cite{chizat2020faster} yields a convergence rate of $\sqrt{K}/\sqrt{n}$, where $K$ is the number of margins and $n$ is the number of samples. This rate is not ideal as $K$ can range from $0$ to $n$. For $K = \sqrt{n}$, the rate becomes $1/n^{1/4}$, and if $K/n$ is constant, the bound is vacuous. Our key technical contribution is removing this dependence on $K$ through a subtle two-step approach, uncovering a surprising property of the optimal coupling for MOT. We subsequently extend the results of \cite{chizat2020faster} to achieve the convergence of $\QMW(\mustar;A)$.

\begin{theorem}
    \label{thm:main-upper-bound}
    For any $n,d,K$, we have
    \begin{align}\label{eq:MOT.expectation.upper.bound}
    \begin{split}
        \revindent[2] \EE\left[\left|\MW_2^2(\muhat_n) - \MW_2^2(\mustar)\right|\right] \lesssim n^{-1/2} + 
    \prn*{1/K} \wedge
    \begin{cases} 
         K^{-1/2} (n/K)^{-2/d}, & \text{if } d > 4, \\
         n^{-1/2} \log(n/K), & \text{if } d = 4, \\
         n^{-1/2}  , & \text{if } d < 4,
        \end{cases}
    \end{split}
    \end{align}
    where the notation $\lesssim$ hides constants that only depend on the dimension $d$.  In particular, for $d\leq 4$, we have 
    \begin{align*}
        \EE\left[\left|\MW_2^2(\muhat_n ) - \MW_2^2(\mustar)\right|\right] \lesssim \wt{O} (n^{-1/2}).
    \end{align*}
\end{theorem}

\begin{proof}[Proof sketch]
The full proof is deferred to \cref{app:theory}. According to the techniques from \cite{chizat2020faster}, to achieve our stated convergence rate, the main challenge is to control $\sup_y\norm{\sum_k y(k)}$ for any point $y=(y(1),...,y(K))$ in the support of the optimal coupling. This is difficult when $\mu_\star(1),...,\mu_\star(K)$ are arbitrary distributions with densities. The optimal coupling will also have a density, and adding any coupled point $y$ with positive density (but zero probability) will not affect the $\ell_2$ cost but will make $\sup_y\norm{\sum_k y(k)}$ uncontrollable. 
We address this difficulty by (1) controlling $\sup_y \norm{\sum_k y(k)}$ in the support of the optimal coupling for discretely supported distributions, and (2) generalizing the convergence rate by bypassing the techniques from \cite{chizat2020faster}.

For discretely supported distributions, the probability mass of any support of the optimal coupling is positive, avoiding issues seen with continuous distributions. If two coupled points $y_1$ and $y_2$ are in the support, any exchange in their arguments will not decrease the coupled $\ell_2$ cost. This is the analogue of the argument for optimal transport, where a change in the optimal transport map will not reduce the transport value. This surprisingly provides control over $\sup_y\norm{\sum_k y(k)}$ for $y$ in the support of the optimal coupling, as shown in \cref{lem:finite-support-eta-bound}. Thus, somewhat unexpectedly, the $\ell_2$ norm minimization forces the optimal coupling to achieve point-wise $\ell_2$ norm guarantees for discretely supported distributions, despite being an average minimization.

The second step is to generalize the convergence to arbitrary distributions, which is challenging as we can't control $\sup_y\norm{\sum_k y(k)}$ generally. To show convergence in total $\ell_2$ cost, we use an imaginary sampling process. Suppose we sample enough points from $\mu_\star(1),...,\mu_\star(K)$ to form empirical distributions $\mu_N(1),...,\mu_N(K)$, with $N$ arbitrarily large. The convergence of $\MW_2^2(\mu_N)$ to $\MW_2^2(\mustar)$ holds with a suboptimal rate by directly applying \cite{chizat2020faster}, but since $N$ is imaginary, we choose it large enough to eliminate the bad dependence on $K$. These empirical distributions are discretely supported, allowing us to achieve the desired rate by sampling from them, albeit with replacement. When $N$ is large enough, sampling $n$ points with replacement approximates sampling $n$ points without replacement (where the latter is equivalent to sampling $n$ points from $\mustar$). Thus concluding our proof.

\end{proof}

\paragraph{General quadratic functions}
The convergence result can be obtained for arbitrary quadratic functions by generalizing the approach from  \cite{chizat2020faster}. 
Concretely, we have the following theorem.
\begin{proposition}
    \label{prop:quadratic-form-upper-bound-simplified}
    For any symmetric matrix $A \in \bR^{K\times K}$, we have
    \begin{align}
    \begin{split}
        \revindent[0] \EE\left[\left|\QMW(\muhat_n;A ) - \QMW(\mustar;A)\right|\right] \\
        &\lesssim 
        K\sqrt{\sum\limits_{k,l}a_{kl}^2}   \cdot  (n)^{-1/2}  
        + \prn*{\sum\limits_{k=1}^{K} |a_{kk}|} \wedge
    \begin{cases} 
        \norm{A}_{1,\infty} K (n/K)^{-2/d}, & \text{if } d > 4, \\
        \norm{A}_{1,\infty} K^{3/2} n^{-1/2} \log(n/K), & \text{if } d = 4, \\
        \norm{A}_{1,\infty} K^{3/2}n^{-1/2}   , & \text{if } d < 4,
        \end{cases}
    \end{split}
    \end{align}
    where the notation $\lesssim$ hides constants that depend on the dimension $d$ and $\norm{A}_{1,\infty} \ldef \sup_{k} \prn*{ \sum_{l}|a_{kl}|} $.
\end{proposition}

\subsection{Lower bound}
\label{sec:lower-bound}

In this section, we show that the minimax lower bound on the convergence for any estimator is $\Omega(n^{-1/2})$ for estimating the MOT problem (Eq.~\eqref{eq:quadratic.objective}). 
This, together with the upper bound shown in \Cref{thm:main-upper-bound}, asserts that when $d\leq 4$, the plug-in estimation is minimax optimal.

\begin{proposition}\label{prop:minimax-lower-bound}
For any $K\geq 2$, any $d\geq 1$, and any estimator $F(\set{Y_i(k)}_{i\in [n/K],k\in [K]})$, we have
\begin{align*}
    \sup_{\mustar}\EE\left[\left| F(\set{Y_i(k)}_{i,k}) - \MW_2^2(\mustar)\right|\right] \gtrsim
    \begin{cases}
        n^{-1/2} & \text{if~}d\leq 4,\\
        K^{-2}(n/K\log (n/K))^{-2/d} \wedge n^{-1/2} & \text{if~}d > 4,
    \end{cases}
\end{align*} 
\end{proposition}

\begin{proof}[Proof sketch]
    The full proof is deferred to \Cref{app:lower-bound}. This proof uses Le Cam's two-point method (Lemma 1 of \cite{yu1997assouad}), where we consider $\mustar(1) =\dots = \mustar(K) = Ber(1/2)$ or $\mustar(1) =\dots = \mustar(K) = Ber(1/2+\veps)$. When $\veps = \Theta(n^{-1/2})$ is small enough, then there is a constant probability any estimator would not be able to differentiate the two cases. 
    Here $Ber(p)$ denotes the Bernoulli distribution with probability parameter $p$. Meanwhile, the divergence between the target values in these two cases scales with $\veps$.  Furthermore, by the lower bound of Wasserstein distance estimation from theorem 22 of \cite{manole2024sharp}, we obtain a lower bound of $K^{-2}(n/K)^{-2/d}$ when $d>4$.
\end{proof}

\section{Numerical experiments}\label{sec:empirical}

\subsection{Implementation}\label{sec:implementation}
There are several algorithms proposed to solve the MOT problem in the literature. For example, Sinkhorn's algorithm \parencite[Section 10.1]{peyre2019computational}, Multi-marginal Sinkhorn (\cite{lin2022complexity}), Primal-Dual Accelerated Alternating Minimization (PD-AAM) (\cite{tupitsa2020multimarginal}), etc. In addition, we also consider the  Multi-marginal Greenkhorn algorithm which is extended from \cite{altschuler2018nearlinear}.
The detail of the Multi-marginal Greenkhorn algorithm and Multi-marginal Sinkhorn algorithm is deferred to Appendix \ref{app:algorithm}.
We provide implementations for all four algorithms, while our numerical results are obtained from MOT Sinkhorn for the following three advantages: (1) This algorithm has a convergence guarantee, in which case we can use the dual objective as a lower bound. We obtained the lower bound for the Epitaxialgrowth dataset in Section \ref{sec:partial.identification} and Section \ref{sec:finite.sample} using this bound.
(2) The algorithm is compatible with the celebrated log sum exponential trick in our implementation to avoid instability due to floating point errors caused by finite machine representation. (3) Even when the algorithm does not converge, during the iteration lower bounds for the target value can be obtained. In particular, the lower bounds in Section \ref{sec:partial.identification} and that in Section \ref{sec:covariance} are obtained using the lower bounds obtained during the iteration without convergence.

\subsection{Synthetic dataset}

\paragraph{1-d Gaussian distributions} We pick $K = 3$ margins, each following a Gaussian distribution with different standard deviations. Here, we generate $200$ data points in each margin ($600$ data points in total), with the standard deviations of the margins being $2$, $0.3$, $0.1$, respectively. The objective function chosen here is $(Y(1) + Y(2) + Y(3))^2/9$. For any $k \in [K]$, if the marginal distributions are Gaussian, i.e. $ \mu(k) = \mathcal{N}(0, \sigma_k^2) $ for $ \sigma_k > 0 $, then we have the following explicit form for the theoretical optimal multi-marginal Wasserstein loss.
\begin{lemma}[\cite{doi:10.1080/03610926.2019.1586937}]
Let $\mu(1)= \mathcal{N}(0, \sigma_1^2),...,\mu(K)=\mathcal{N}(0, \sigma_K^2)$ be $K$ Gaussian distributions. Then we have
\begin{align*}
    \MW_2^2(\mu(1), \ldots, \mu(K)) = \frac{1}{K^2} \left( \left( 2 \max_{k \in [K]} \sigma_k - \sum_{l=1}^K \sigma_l \right) \vee 0 \right)^2.
\end{align*}
\end{lemma}
Using the provided standard deviation values from the data generation process, the corresponding theoretical optimal loss can be calculated as $0.28$. The empirical MOT lower bound we get by running our Sinkhorn algorithm is $0.24$. We note that the empirical optimal value we obtained serves as a lower bound for the theoretical optimal value.

\paragraph{2-d Gaussian distributions with known data generating distribution} We consider $3$ margins, each with a $2$-dimensional outcome. We generate $100$ data points $\set{(Y_i^1(1), Y_i^2(1),Y_i^1(2), Y_i^2(2),Y_i^1(3), Y_i^2(3))}_{i\in [100]}$ from the 6-d Gaussian distribution $\cN(\mathbf{0}, \diag( 2,2,0.3,0.3,0.1,0.1))$. Then, we separate each data point $i$ into three margins, each being 2-dimensional, that is $(Y_i^1(1), Y_i^2(1))$, $(Y_i^1(2), Y_i^2(2))$, and $(Y_i^1(3), Y_i^2(3))$.
The objective we choose in this case is $\sum_{j \in \{1,2\}} (Y^j(1) - Y^j(2)/2 - Y^j(3)/2)^2$.

The exact expression for the theoretical optimal objective function is not accessible for $d > 1$. Therefore, we compare the empirical optimal objective to the known true value to assess whether it serves as a valid lower bound. The theoretical expectation of the objective is calculated to be $4.20$, while the empirical optimal objective, obtained by running our Sinkhorn algorithm, is $1.78$, confirming that the empirical optimal objective is indeed consistent with the true value of the causal estimand.

\subsection{Real data applications}

\subsubsection{Detection of treatment effect}\label{sec:partial.identification}

Continued from \Cref{eg:interaction-effect}, we compute tight lower bounds for causal estimands taking the form $\EE_{\mustar}[(\sum_{k\in[K]}\beta_kY_i(k))^2]$ on three real datasets (results are summarized in Table \ref{tb:ht}). We explain the three experiments in detail below.
For details of the datasets, please refer to \Cref{sec:datasets}.

We first consider the estimand $\EE_{\mustar}[(Y_i(1,1) - Y_i(1,0) - Y_i(0,1) + Y_i(0,0))^2]$ related to interaction effects.
We use the Epitaxial Layer Growth data  \parencite{wu2011experiments}, which investigates the impact of experimental factors on the thickness of growing an epitaxial layer on polished silicon wafers. 
In particular, we focus on the interaction of two binary treatments: susceptor's rotation and nozzle position, and there are 6 data points per treatment combination.
Using the Multimarginal Sinkhorn algorithm with error margin $\epsilon = 0.001$, the algorithm converges, and we get the lower bound to be $0.0460$. 
We remark that the lower bound we get is $102.8 \%$ higher than the baseline lower bound.

Next, we consider the estimand $\EE_{\mustar}[\|Y_i(1,0)/2 + Y_i(0,1)/2 - Y_i(0,0)\|_2^2]$ related to contrast effects.
We use the Helpfulness data simulated from \cite{jia2022empathy}, which investigates factors influencing individuals' empathy.
In particular, we focus on two binary treatments: whether the individual has diverse experiences (i.e., has experienced both experimental conditions analogous to ``wealth'' and ``poverty'') and whether the individual was provided with motivations for being empathetic.
We consider the margin $(0,0)$ as the control level ($141$ units), $(0,1)$, $(1,0)$ as two treatment levels ($96$, $41$ units, respectively), and adopt the contrast effect $Y_i(1,0)/2 + Y_i(0,1)/2 - Y_i(0,0)$. 
Using the Multimarginal Sinkhorn algorithm with error margin $\epsilon = 0.001$, the algorithm converges, and we get the lower bound of the variance to be $0.432$. 
We remark that the lower bound we get is also $79.5 \%$ higher than the baseline lower bound.

Finally, we consider the estimand $\EE_{\mustar}[||Y_i(1,0)/2 + Y_i(0,1)/2 - Y_i(0,0)||_2^2]$ for the Education data\footnote{We focus on students whose baseline GPA falls within the lower $20\%$.} from the Student Achievement and Retention (STAR) Demonstration Project \parencite{angrist2009incentives}.
Here researchers investigate the impact of two scholarship incentive programs, the Student Support Program (SSP) and the Student Fellowship Program (SFP), to the academic performance.
We consider the margin $(0,0)$ as the control level ($161$ units), $(0,1)$, $(1,0)$ as two treatment levels ($37$, $36$ units, respectively), and adopt the contrast effect $Y_i(1,0)/2 + Y_i(0,1)/2 - Y_i(0,0)$. 
Using the Multimarginal Sinkhorn algorithm with error margin $\epsilon = 0.001$, the algorithm converges, and we get the lower bound of the variance to be $0.0654$. 
We remark that the lower bound we get is also $72.3 \%$ higher than the baseline lower bound.

\begin{table}[ht]
\centering
\caption{Lower bounds for detection of treatment effect.} \label{tb:ht}
\begin{tabular}{lccc}
\hline

Dataset & Baseline lower bound & MOT lower bound & Percentage improvement \\
\hline
Epitaxial Layer Growth & 0.0227 &  0.0460 & 102.8\% \\
Helpfulness & 0.241 & 0.432 & 79.5\% \\
Education & 0.0380 & 0.0654 & 72.3\% \\
\hline
\end{tabular}
\end{table}

\subsubsection{Covariance between treatment effects}\label{sec:covariance}

Continued from \Cref{eg:cov-between-te}, we investigate the correlation between treatment effects regarding different outcomes.
In particular, we return to the Education dataset of the STAR project in \Cref{sec:partial.identification}.
We consider the contrast effect in \Cref{sec:partial.identification} for the responses $Y^j(\cdot)$, $j \in \{1,2\}$ representing the $j$-th year GPA at different treatment levels. 
We estimate the partially identified set for the covariance between the contrasts regarding the first-year GPA and the second-year GPA.

The lower bound of the covariance is $-0.322$, and the upper bound of the covariance is $2.256$. 
Even though we can not rule out the possibility that the treatment effects in the first-year GPA and the second-year GPA are negatively correlated, the asymmetry $|-0.322| \ll |2.256|$ seems to indicate that the treatment effects are positively correlated.
To sharpen the bound of the covariance, one can further adjust for baseline covariates like initial GPA and age to remove the variation in the first, and second-year GPAs that are not attributed to the STAR project (discussed in \Cref{sec:discussion}).

There are five outcomes available in the dataset: the GPA at the end of year 1 (year 2), whether the student is in good standing at the end of year 1 (year 2), and the grade of the first semester in the first year. 
We also carried out the computation for covariances for other outcome pairs.
The results of lower and upper bounds are displayed in \Cref{fig:cov_bounds}.

\begin{figure}[h]
    \centering
    \includegraphics[width = \textwidth]{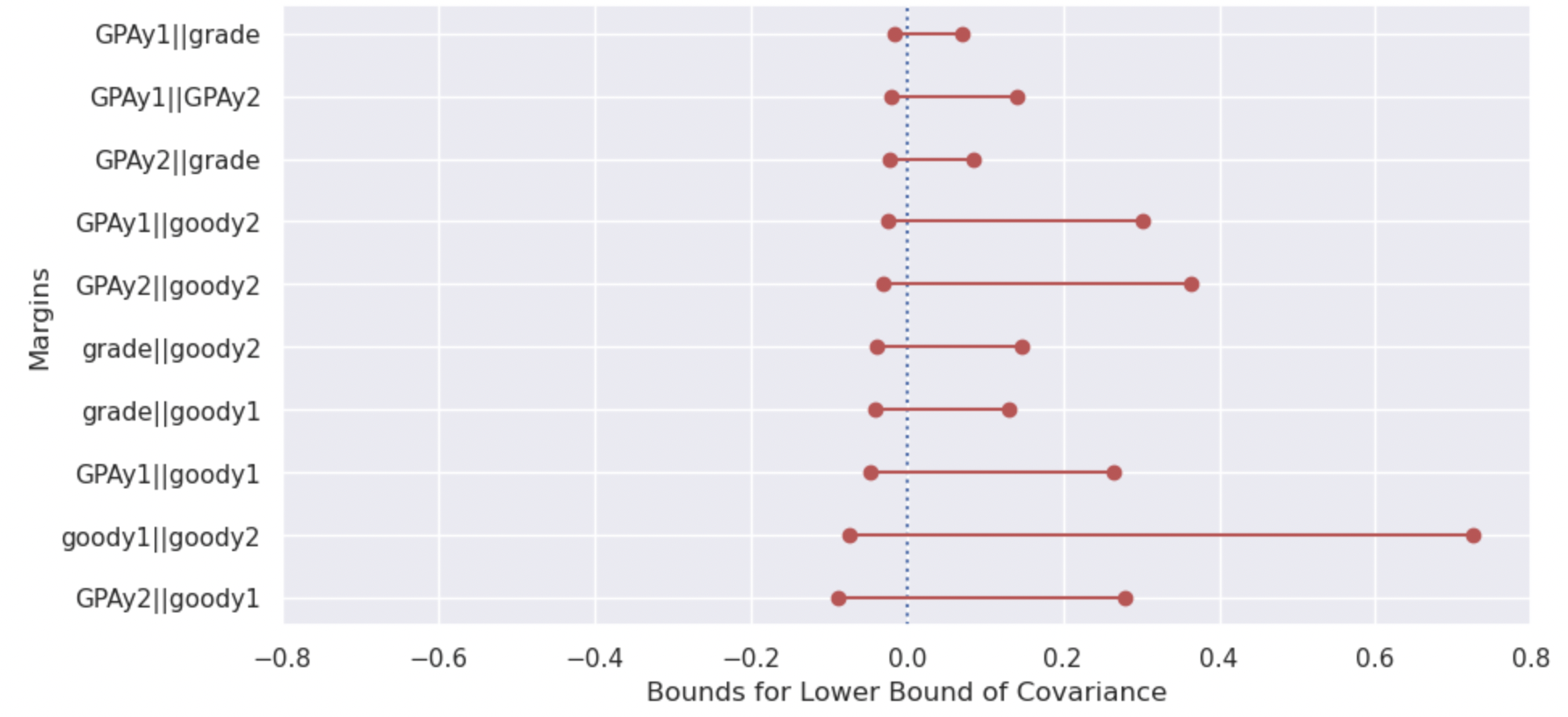}
    \caption{Lower and upper bounds for the covariance between different pairs of outcomes. 
    Here GPAy1 means the GPA at the end of year 1, and similarly for GPAy2; goody1 means whether the student is in good standing at the end of year 1, and similarly for goody2; grade means the grade of the first semester in the first year.}
    \label{fig:cov_bounds}
\end{figure}

\subsubsection{Neymanian confidence interval}\label{sec:finite.sample}

Continued from \Cref{eg:design-based-inference}, we apply the proposed method to tighten the conventional estimator of the variance \eqref{eq:variance.design.based}. 
We return to the Expitaxial layer growth dataset.
In particular, we focus on the three treatment levels $(0,1)$, $(1,0)$, $(1,1)$ and consider the contrast $(Y_i(1,1) - Y_i(1,0)/2 - Y_i(0,1)/2)$.
Our lower bound of $S_{\tau}^2$ in \eqref{eq:variance.design.based} is noticeably greater than zero, indicating that bounding it by zero is overly conservative. 
As a result, our estimated variance has decreased the conventional counterpart by a significant $8.16$ percent.
For sample size calculation, a reduction of $8.16$ percent in variance equates to a saving of $8.16$ percent in the number of units/repeats required. 
The results are shown in \Cref{tb:var&ci}.

\begin{table}[ht]
\centering
\caption{Neymanian confidence intervals. The contrast effect estimate is $0.0533$. } \label{tb:var&ci}
\begin{tabular}{lcc}
\hline
Method & Variance estimate  & 95\% Confidence interval for $\tau_N$ \\
\hline
Convention 
& 0.113 & $(-0.609, 0.711)$ \\
Our proposal
& 0.104 & $(-0.581, 0.684)$ \\
\hline
\end{tabular}
\end{table}

\section{Discussion}\label{sec:discussion}

In this work, we advocate for using MOT to obtain identified sets for causal estimands that are only partially identifiable. 
We perform theoretical analysis of the plug-in estimator of the identified sets, provide off-the-shelf implementation, and demonstrate the potentials of the approach for various causal tasks across multiple datasets.

We provide several future research directions.
\begin{itemize}
    \item Covariate-assisted MOT.
    In the presence of covariates $X$, the conditional marginal distributions $Y(k) \mid X$ should be respected, a more stringent condition compared to preserving unconditional marginal distributions. 
    We denote the set of joint couplings that adhere to these conditional marginals as $\Gamma_X$, then $\Gamma_X \subseteq \Gamma$ and yields larger population lower bounds of the objectives~\eqref{eq:quadratic.objective} and \eqref{eq:quadratic.type} with $\Gamma$ replaced by $\Gamma_X$.
    For empirical lower bound, conditioning on more covariates will decrease the estimation accuracy of the plug-in estimator. Striking a balance between the population's lower bound and the empirical performance by selecting a subset of covariates 
    and solving the associated MOT problem is of interest.

    \item Approximate joint mixability. For any $K$ centered distributions $\mu(1),...,\mu(K)$ on $\bR^d$, they are jointly mixable if there exists a coupling such that for any $\bm Y(\cdot)$ on the support of the coupling, $\sum_{k=1}^K Y(k)/K = 0$~\parencite{wang2016joint}. 
    The objective in Eq.~\eqref{eq:quadratic.objective} can be used to define a generalization of the joint mixability, where the level of joint mixability is evaluated by the 2-norm. 
    One interesting phenomenon as an implication of \Cref{lem:finite-support-eta-bound} is that the 2-norm approximate joint-mixability implies approximate joint-mixability in infinity norm for discretely supported distributions. Concretely, for any $\bm Y(\cdot)$ on the support of the optimal coupling in terms of objective in Eq.\eqref{eq:quadratic.objective}, we have $|\sum_{k=1}^K Y(k)/K| \lesssim 1/\sqrt{K}$ almost surely. This posits a promising research direction of approximate joint mixability. Many open problems can be considered, e.g., can we extend \Cref{lem:finite-support-eta-bound} to arbitrary distributions? What are the relationships between the $p$-norm approximate joint mixability for $p\geq 1$?

    \item Implementation. 
    There is a tradeoff between precision and the run time/resources. Especially when the accuracy tolerance $\varepsilon$ is small, the algorithm might become unstable due to floating point errors. Also, the run time will be much higher.
    Compared to OT, MOT requires a higher resolution, or equivalently a smaller tolerance parameter $\varepsilon$, due to a larger number (exponential to the number of margins) of cell probabilities, making the acceleration of even higher importance.
    
    Future work could explore potential ways to achieve higher precision in a more efficient way.
    One possibility is utilizing a low-rank tensor approximation to the quadratic objective function. 
    In addition, in the current implementation, we use greedy Sinkhorn where the parameters are updated every $K$ iterations (after we calculate all $K$ margins). 
    We have seen numerical success in simply updating one margin per iteration in a round-robin fashion, which is computationally economic. However, the convergence of this algorithm remains to be studied.

\end{itemize}

\section*{Acknowledgements}
We thank Gabriel Peyr\'{e} for the useful discussion on the implementation of MOT algorithms. Jian Qian acknowledges support from ARO through award W911NF-21-1-0328 and from the Simons Foundation and NSF through award DMS-2031883. Shu Ge acknowledges the support of ARO through award W911NF-21-1-0328.

\printbibliography

\appendix

\section{Notations}
\label{app:notation}

\begin{table}[H]
    \centering
    \caption{Notation Table}
    \begin{tabular}{m{5.5cm} m{9.5cm}}
        \toprule
        \textbf{Notation} & \textbf{Description} \\
        \midrule
        $n$ & Total number of samples \\
        $d$ & Dimension of the outcome space \\
        $K$ & Number of possible treatment levels \\
        $\bY_i(\cdot)$ & Potential outcomes for unit $i$ \\
        $Y_i(k)$ & Potential outcome for unit $i$ at treatment level $k$ \\
        $Y_i^j(k)$ & $j$-th dimension of potential outcome for unit $i$ at treatment level $k$ \\
        $Y_i(0,0)$, $Y_i(0,1)$, $Y_i(1,0)$, $Y_i(1,1)$ & Potential outcomes under different binary treatment combinations \\
        $\var_{\mustar}(Y_i(2) - Y_i(1))$ & Variance of the treatment effect \\
        $\cov_{\mustar}(Y_{i}^1(2) - Y_{i}^1(1), Y_{i}^2(2) - Y_{i}^2(1))$ & Covariance between individual treatment effects across different dimensions \\
        $\Gamma$ & Collection of joint probabilities respecting the marginal distributions \\
        $\mustar$ & Unknown joint distribution of potential outcomes \\
        $\bar{\mu}_\star$ & Centered counterpart of ${\mu}_\star$, that is, $\bar{\mu}_\star(k)$ is ${\mu}_\star$ minus its expectation \\
        $\mustar(k)$ & Marginal distribution of $Y_i(k)$ \\
        $\muhat_n$ & Empirical distribution of $\mustar$ with $n$ samples \\
        $W_i$ & Realized treatment level for unit $i$ \\
        $\mathbb{E}_{\mustar}[\ell(\bm{Y}_i(\cdot))]$ & Expectation of a causal estimand $\ell$ under the distribution $\mustar$ \\
        $\QMW(\cdot; A)$ & Quadratic objective parameterized by $A$ \\
        $\MW_2^2(\cdot)$ & Multi-Marginal Wasserstein Loss: Average squared $\ell_2$ norm of potential outcomes \\
        $\boldsymbol{\beta}$ & Contrast vector for treatments with multiple levels \\
        $\tau_{\boldsymbol{\beta}}$ & Average contrast effect estimand \\
        $\hat{\tau}_{\boldsymbol{\beta}}$ & Difference-in-means type estimator \\
        $S_k^2$ & Sample variance of potential outcomes at treatment level $k$ \\
        $S_{\tau}^2$ & Sample variance of the contrast treatment effect \\
        $V_{\boldsymbol{\beta}}$ & Variance of the difference-in-means type estimator \\
        \bottomrule
    \end{tabular}
    \label{tab:notation}
\end{table}

\section{Technical lemmas}

Throughout the appendix, for any $L>0$, let $\Flip$ denote the $L$-Lipschitz function class supported on $B_1^d(0)$ such that all functions $\phi\in\Flip$ has $\phi(0) = 0$.
Denote by $\gmstar_A(\mu(1),\dots,\mu(K))$ the optimal coupling between $\mu(1),\dots,\mu(K)$ in terms of the objective \pref{eq:quadratic.type}.
Denote by $\gmstar(\mu(1),\dots,\mu(K))$ the optimal coupling between $\mu(1),\dots,\mu(K)$ in terms of the objective \pref{eq:quadratic.objective}.

\begin{lemma}
\label{lem:radermacher-complexity}
    For any $\mu$ supported on $B_1^d(0)$, the Radermacher complexity of $\cF_L$ under $\mu$ is bounded by
    \begin{align*}
        \cR_n(\cF_L,\mu)\ldef \En_{\eps, X} \brk*{ \sup_{\phi\in \cF_L} \abs*{ \frac{1}{n} \sum_{i=1}^n \eps_i \phi(Y_i)  } } \lesssim  
        \begin{cases}
            Ln^{-2/d}&\text{if~}d>4,\\
            Ln^{-1/2}\log n &\text{if~}d=4,\\
            Ln^{-1/2} & \text{if~}d<4,
        \end{cases}
    \end{align*}
    where $\eps = (\eps_1,...\eps_n)$, $\eps_1,...,\eps_n$ are independent Radermacher variables and $Y_1,...,Y_n$ are independent random variables sampled from the distribution $\mu$.
\end{lemma}

\begin{proof}[\pfref{lem:radermacher-complexity}]
We follow the derivation of \cite[Lemma 4]{chizat2020faster}.
On $B_1^d(0)$, any function $\phi\in \cF_L$ is bounded by $\norm{\phi}_\infty \lesssim L$. Thus by \cite[Theorem 5.22]{wainwright2019high}, we have
\begin{align*}
    \cR_n(\cF_L,\mu) \lesssim \inf_{\delta>0} \prn*{\delta + n^{-1/2}\int_\delta^L \sqrt{\log \cN(\cF_L,\norm{\cdot}_\infty,u) d u }  },
\end{align*}
where $\cN(\cF_L,\norm{\cdot}_\infty,u)$ denotes the covering number of $\cF_L$ in $\norm{\cdot}_\infty$ at scale $u$. Then by \cite[Theorem 1]{guntuboyina2012l1} which obtains $\log\cN(\cF_L,\norm{\cdot}_\infty,u)\lesssim (u/L)^{-d/2}$, we have
\begin{align*}
    \cR_n(\cF_L,\mu) &\lesssim \inf_{\delta>0} \prn*{\delta + n^{-1/2}\int_\delta^L (u/L)^{-d/4} d u }  \\
    &= \inf_{\delta>0} \prn*{\delta + L n^{-1/2}\int_{\delta/L}^1 u^{-d/4} d u }\\
    &\lesssim 
        \begin{cases}
            Ln^{-2/d}&\text{if~}d>4,\\
            Ln^{-1/2}\log n &\text{if~}d=4,\\
            Ln^{-1/2} & \text{if~}d<4,
        \end{cases}
\end{align*}
where for $d< 4$, we choose $\delta=0$ and $\delta=L n^{-1/2}$ for $d=4$ and $\delta =L n^{-2/d}$ for $d>4$. 
\end{proof}

\begin{lemma}
\label{lem:expectation-radermacher}
     For any $\mu$ supported on $B_1^d(0)$ and $\muhat_n$ an empirical distribution of $n$ independent samples, we have
    \begin{align*}
        \En \brk*{\sup_{\phi \in \mathcal{F}_{L}} \abs*{ \int \phi(y) d({\muhat}_n - \mu)}} \lesssim 2 \cR_n(\cF_L,\mu) + Ln^{-1/2}.
    \end{align*}
\end{lemma}

\begin{proof}[\pfref{lem:expectation-radermacher}]
By \cite[Theorem 4.10]{wainwright2019high}, we have for all $\delta>0$, with probability at least $1-\exp( -n\delta^2/(2L^2))$,
    \begin{align*}
        \sup_{\phi \in \mathcal{F}_{L}} \abs*{ \int \phi(y) d({\muhat}_n - \mu)} \leq 2 \cR_n(\cF_L,\mu) + \delta.
    \end{align*}
Then since $\En[X] = \int_0^\infty \bP(X\geq \delta)\d \delta$ for all positive random variable, we have
    \begin{align*}
                \revindent\En \brk*{\sup_{\phi \in \mathcal{F}_{L}} \abs*{ \int \phi(y) d({\muhat}_n - \mu)} -  2 \cR_n(\cF_L,\mu)} \\
                &\leq \En \brk*{\prn*{\sup_{\phi \in \mathcal{F}_{L}} \abs*{ \int \phi(y) d({\muhat}_n - \mu)} -  2 \cR_n(\cF_L,\mu)}_+} \\
                &= \int_0^\infty \bP \prn*{\prn*{\sup_{\phi \in \mathcal{F}_{L}} \abs*{ \int \phi(y) d({\muhat}_n - \mu)} -  2 \cR_n(\cF_L,\mu)}_+\geq \delta} \d\delta \\
                &\leq \int_0^\infty \exp( -n\delta^2/(2L^2)) \d\delta \\
                &\lesssim L n^{-1/2}.
    \end{align*}
\end{proof}

\section{Omitted proofs from \Cref{sec:setup}}

\begin{lemma}
    \label{lem:mean-decomposition}
For any $K$ distributions $\mu(1),\dots,\mu(K)$ on $\bR^d$, let $\mubar(1),\dots,\mubar(K)$ be the centered counterpart, that is, $\mubar(k)$ is a translation of $\mu(k)$ by $-\En_{\mu(k)}\brk*{X}$ and $Y_1(k)\sim \mu(k)$ for $k\in [K]$. Then we have
\begin{align*}
\MW_2^2(\mu(1),\dots,\mu(K)) = \MW_2^2(\mubar(1),\dots,\mubar(K)) + \left\| \frac{\sum_{k=1}^{K} \EE_{\mu(k)} [Y_1(k)]}{K}  \right\|_2^2
\end{align*}
\end{lemma}

\begin{proof}[\pfref{lem:mean-decomposition}]
Suppose $\gmstar = \gmstar(\mu(1),...,\mu(K))$ achieves $\MW_2^2(\mu(1),\dots,\mu(K))$, then we have
\begin{align*}
    &\MW_2^2(\mu(1),\dots,\mu(K)) \\
    &= \int_{y} \norm*{ \frac{\sum_{k=1}^K y(k)}{K} }^2 \gmstar(d y)\\
    &= \int_{y} \norm*{ \frac{\sum_{k=1}^K \prn*{y(k)-\En_{\mu(k)} \brk*{Y_1(k)} } }{K} }^2 \gmstar(d y) +  \left\| \frac{\sum_{k=1}^{K} \EE_{\mu(k)} \brk*{Y_1(k)}}{K}  \right\|_2^2\\
    &\leq \MW_2^2(\mubar(1),\dots,\mubar(K)) +   \left\| \frac{\sum_{k=1}^{K} \EE_{\mu(k)} \brk*{Y_1(k)}}{K}  \right\|_2^2.
\end{align*}
The inequality can be reversed considering the optimal coupling $\bar{\gamma}^\star = \gmstar(\mubar(1),...,\mubar(K))$, which concludes our proof. 
\end{proof}

\begin{lemma}
    \label{lem:mean-decomposition-quadratic}
For any $K$ distributions $\mu(1),\dots,\mu(K)$ on $\bR^d$, let $\mubar(1),\dots,\mubar(K)$ be the centered counterpart, that is, $\mubar(k)$ is a translation of $\mu(k)$ by $-\En_{\mu(k)}\brk*{X}$ and $Y_1(k)\sim \mu(k)$ for $k\in [K]$.  Let $A$ be any symmetric matrix.
Then we have
\begin{align*}
\QMW(\mu(1),\dots,\mu(K);A) = \QMW(\mubar(1),\dots,\mubar(K);A) + \sum_{k,l }a_{kl}\inner*{\EE_{\mu(k)} [Y_1(k)] , \EE_{\mu(l)} [Y_1(l)]}.
\end{align*}
\end{lemma}
\begin{proof}[\pfref{lem:mean-decomposition-quadratic}]
Suppose $\gmstar = \gmstar(\mu(1),...,\mu(K);A)$ achieves $\QMW(\mu(1),...,\mu(K);A)$, then we have
\begin{align*}
    &\QMW(\mu(1),\dots,\mu(K);A) \\
    &= \int_{y}  \sum_{k,l} a_{kl}\inner*{y(k) , y(l)} \gmstar(d y)\\
    &=  \int_{y}  \sum_{k,l} a_{kl}\inner*{y(k)- \EE_{\mu(k)} [Y_1(k)] , y(l)-\EE_{\mu(l)} [Y_1(l)]} \gmstar(d y) + \sum_{k,l }a_{kl}\inner*{\EE_{\mu(k)} [Y_1(k)] , \EE_{\mu(l)} [Y_1(l)]}\\
    &\leq \QMW(\mubar(1),\dots,\mubar(K);A) +   \sum_{k,l }a_{kl}\inner*{\EE_{\mu(k)} [Y_1(k)] , \EE_{\mu(l)} [Y_1(l)]}.
\end{align*}
The inequality can be reversed considering the optimal coupling $\bar{\gamma}^\star = \gmstar(\mubar(1),...,\mubar(K))$, which concludes our proof. 
\end{proof}

\section{Omitted proofs from \Cref{sec:theory}}
\subsection{Upper bound}
\label{app:theory}

\begin{lemma}
    \label{lem:mean-concentration} For any matrix $A = (a_{ij})_{1\leq i,j\leq K}$, we have
    \begin{align*}
        \En\brk*{  \abs*{ \sum_{k,l} a_{kl}\inner*{  \En_{\mustar(k)} Y_1(k),   \En_{\mustar(l)} Y_1(l)}   - \sum_{k,l} a_{kl}  \inner*{\En_{\muhat_n(k)} Y_1(k), \En_{\muhat_n(l)} Y_1(l)}   }} \lesssim  K\sqrt{ \sum_{k,l} a_{kl}^2 }\cdot n^{-1/2} .
    \end{align*}
\end{lemma}

\begin{proof}[\pfref{lem:mean-concentration}]
We first separate the quantity
\begin{align*}
    &\sum_{k,l} a_{kl}\inner*{  \En_{\mustar(k)} Y_1(k),   \En_{\mustar(l)} Y_1(l)}   - \sum_{k,l} a_{kl}  \inner*{\En_{\muhat_n(k)} Y_1(k), \En_{\muhat_n(l)} Y_1(l)}  \\
    & =  \sum_{k=1}^K a_{kk} \prn*{\inner*{\En_{\mustar(k)}[Y_1(k)], \En_{\mustar(k)}[Y_1(k)]  }- \inner*{\En_{\muhat(k)}[Y_1(k)], \En_{\muhat(k)}[Y_1(k)]  }}\\
    &\quad+  \sum_{k\neq l} a_{kl} \inner*{\En_{\mustar(k)}[Y_1(k)], \En_{\mustar(l)}[Y_1(l)]  } \\
    &\quad- \sum_{k\neq l} a_{kl} \inner*{\En_{\muhat_n(k)}[Y_1(k)], \En_{\muhat_n(l)}[Y_1(l)]  } \\
    &=\sum_{k=1}^K a_{kk} \inner*{ \En_{\mustar(k)} Y_1(k) -  \En_{\muhat_n(k)} Y_1(k) ,  \En_{\mustar(k)} Y_1(k)   + \En_{\muhat_n(k)} Y_1(k)  } \\
    &\quad+  \sum_{k\neq l} a_{kl} \inner*{\En_{\mustar(k)}[Y_1(k)], \En_{\mustar(l) - \muhat_n(l)}[Y_1(l)]  } \\
    &\quad+ \sum_{k\neq l} a_{kl} \inner*{\En_{\mustar(k)-\muhat_n(k)}[Y_1(k)], \En_{ \muhat_n(l)}[Y_1(l)]  }.
\end{align*}
Then, by triangle inequality together with Khintchine's inequality (Exercise 2.6.7 of \cite{vershynin2018high}), we have
\begin{align*}
    &   \En\brk*{  \abs*{ \sum_{k,l} a_{kl}\inner*{  \En_{\mustar(k)} Y_1(k),   \En_{\mustar(l)} Y_1(l)}   - \sum_{k,l} a_{kl}  \inner*{\En_{\muhat_n(k)} Y_1(k), \En_{\muhat_n(l)} Y_1(l)}   }}  \\
    &\leq  \En\brk*{ \abs*{ \sum_{k=1}^K a_{kk} \inner*{ \En_{\mustar(k)} Y_1(k) -  \En_{\muhat_n(k)} Y_1(k) ,  \En_{\mustar(k)} Y_1(k)   + \En_{\muhat_n(k)} Y_1(k)  } } } \\
    &\quad + \sum_{k=1}^K \En\brk*{ \En_{\mustar(k)} \abs*{\En_{(\mustar - \muhat_n)(-k)} \brk*{ \sum_{l\neq k} a_{kl} \inner*{ Y_1(k),   Y_1(l) }  }   }  }  \\
    &\quad +  \sum_{k=1}^K \En\brk*{ \En_{\muhat_n(k)} \abs*{\En_{(\mustar - \muhat_n)(-k)} \brk*{ \sum_{l\neq k} a_{kl} \inner*{ Y_1(k),   Y_1(l) }  }   }  } \\
    &\lesssim \prn*{\sqrt{\sum_{k=1}^K a_{kk}^2 } +  \sum_{k=1}^K\sqrt{\sum_{l\neq k} a_{kl}^2}} (n/K)^{-1/2}\\
    &\lesssim \sqrt{K \cdot \sum_{k,l} a_{kl}^2 }\cdot (n/K)^{-1/2} ,
\end{align*}
where the last inequality is by Cauchy-Shwartz.
\end{proof}

\begin{proposition}
    \label{prop:quadratic-form-upper-bound}
    For any symmetric matrix $A \in \bR^{K\times K}$, we have
    \begin{align}
    \begin{split}
        \revindent[0] \EE\left[\left|\QMW(\muhat_n(1),\dots \muhat_n(K);A ) - \QMW(\mustar(1),\dots,\mustar(K);A)\right|\right] \\
        &\lesssim K\sqrt{\sum\limits_{k=1}^{K}a_{kk}^2}   \cdot  n^{-1/2}  + \prn*{\sum\limits_{k=1}^{K} |a_{kk}|} \wedge
    \begin{cases} 
        \eta_{K,d,A}\cdot K^{3/2} (n/K)^{-2/d}, & \text{if } d > 4, \\
        \eta_{K,d,A}\cdot K^2 n^{-1/2} \log(n/K), & \text{if } d = 4, \\
        \eta_{K,d,A}\cdot K^2n^{-1/2}   , & \text{if } d < 4,
        \end{cases}
    \end{split}
    \end{align}
    where the notation $\lesssim$ hides constants that only depend on the dimension $d$. Moreover, the coefficient $\eta_{K,d,A}$ is defined as
    \begin{align*}
    \eta_{K,d,A} \ldef \sup_{\substack{\mu(1),\dots,\mu(K)\\\text{with zero mean} }} \hspace{3pt} \sup_{\substack{(y(1),\dots,y(K))\in\\ \supp \prn*{\gmstar_A(\mu(1),\dots,\mu(K)) }} } \sup_{k\in [K]} \norm*{\sum\limits_{l\neq k}  \frac{a_{kl}y(l)}{\sqrt{K}}  },
    \end{align*}
    where $\mu(1),\dots,\mu(K)$ are $d$-dimensional distributions with zero mean and $(y(1),\dots,y(K))$ in the support of the coupling $\gmstar_A(\mu(1),\dots,\mu(K))$.

\end{proposition}

\begin{proof}[\pfref{prop:quadratic-form-upper-bound}]
    By \cref{lem:mean-decomposition-quadratic,lem:mean-concentration}, we know that the objective can be separated and the mean of the coupling concentrates. Thus, without loss of generality, in this proof, we only consider the marginals $\mustar(1),\dots,\mustar(K),\muhat_n(1),\dots,\muhat_n(K)$ that have zero means. For any $K$ marginals $\mu(1),\dots,\mu(K)$ supported on $B_1^d(0)\subset \bR^d$, we have
    \begin{align*}
        \revindent[.5]\QMW(\mu(1),\dots,\mu(K);A) \\
        &\ldef \inf_{\gamma \in \Gamma(\mu(1),\dots,\mu(K))}  
        \int \sum\limits_{i,j}  a_{ij} y(i)^\top  y(j) ~\gamma(dy)\\
        &\leq  \sum\limits_{k=1}^{K}a_{kk} \int \norm{y(k)}^2 \mu(k,dy(k)) + \sum\limits_{i\neq j} a_{ij} \int y(i)^\top y(j) ~\mu(i,dy(i)) \mu(j,dy(j)) \\
        &= \sum\limits_{k=1}^{K} a_{kk}\int \norm{y(k)}^2 \mu(k,dy(k))\leq \sum\limits_{k=1}^{K} |a_{kk}|,
    \end{align*}
    where the first inequality is by choosing the independent coupling $\gamma(dy) = \prod_{k=1}^{K}\mu(k,dy(k))$ and the last inequality is by the boundedness of $y$. This implies 
    \begin{align}
    \label{ineq:random-control}
        \abs*{\QMW(\mustar(1),\dots,\mustar(K);A) - \QMW(\muhat_n(1),\dots,\muhat_n(K);A)} \leq \sum\limits_{k=1}^{K} |a_{kk}|.
    \end{align}
    Now by the classic duality of \cite{kellerer1984duality}, we have
    \begin{align*}
    \QMW(\mu(1),\dots,\mu(K);A) = \sup_{\substack{f_1,\dots,f_K:\\ \sum\limits_{k=1}^{K}f_k(y(k)) \leq y^\top A y }} \sum\limits_{k=1}^{K} \int f_k(y(k)) \mu(k,dy(k)), 
    \end{align*}
    where $\mu(k,dy(k))$ denote the marignal distribution of $\mu(k)$ on $y(k)$.
    We first separate the square terms out, since they only concern one marginal as
    \begin{align}
        \revindent[0]\sup_{\substack{f_1,\dots,f_K:\\ \sum\limits_{k=1}^{K}f_k(y(k)) \leq y^\top A y }} \sum\limits_{k=1}^{K} \int f_k(y(k)) \mu(k,dy(k)) \notag\\
        &= \sum\limits_{k=1}^{K} \int a_{kk} \norm{y(k)}^2 \mu(k,dy(k)) +  \sup_{\substack{f_1,\dots,f_K:\\ \sum\limits_{k=1}^{K}f_k(y(k)) \leq y^\top A y }} \sum\limits_{k=1}^{K} \int f_k(y(k)) -a_{kk} \norm{y(k)}^2  \mu(k,dy(k))\notag\\
        &=  \sum\limits_{k=1}^{K} \int a_{kk} \norm{y(k)}^2 \mu(k,dy(k)) +  \sup_{\substack{\phi_1,\dots,\phi_K:\\ \sum\limits_{k=1}^{K}\phi_k(y(k)) \leq \sum\limits_{i\neq j} a_{ij}y^\top(i)y(j) }} \sum\limits_{k=1}^{K} \int \phi_k(y(k)) \mu(k,dy(k)),  \label{eq:duality}
    \end{align} 
where the second equality is by setting $\phi_k(y(k)) = f_k(y(k)) - a_{kk} \norm{y(k)}^2$.
Now we are going to show that there exists a sequence of supremum achieving functions $\phistar_{1},\dots,\phistar_K$ that all lie in the space of $\Flip[2K]$. To start with, we note again by the classical duality of \cite{kellerer1984duality} that there exists $\phitil_{1},\dots,\phitil_K$ that achieves the supremum in \cref{eq:duality}. For $k=1,\dots,K$, we iteratively define the function $\phistar_k$ as 
\begin{align*}
\phistar_k(y) \ldef \inf_{\substack{y(1),\dots,y(K)\\
\in \supp(\gmstar_A(\mu(1),\dots,\mu(K)))}} 2\inner*{y, \sum\limits_{l\neq k}a_{kl} y(l) } + \sum\limits_{\substack{i,j:i\neq j,\\ i\neq k,j\neq k}} a_{ij}\inner{y(i), y(j)}  -   \sum_{l=1}^{k-1} \phistar_l(y(l)) - \sum_{l=k+1}^{K} \phitil_l(y(l)).
\end{align*}
And prove by induction that $\phistar_1,\dots,\phistar_{k},\phitil_{k+1},\dots,\phitil_{K}$ are admissible and achieve the supremum in \cref{eq:duality}. When $k=0$, the claim is true by definition. Now suppose the claim is true up to $k-1$.
For any $k$, by definition, we have for any $y(1),\dots, y(K)$
\begin{align*}
    \sum_{l=1}^{k} \phistar_l(y(l)) + \sum_{l=k+1}^{K} \phitil_l(y(l)) \leq \sum\limits_{i\neq j} a_{ij}y(i)y(j).
\end{align*}
Thus $\phistar_1,\dots,\phistar_k,\phitil_{k+1},\dots,\phitil_{K}$ are admissible function for the supremum in \cref{eq:duality}. Furthermore, by the induction hypothesis, $\phistar_1,\dots,\phistar_{k-1},\phitil_{k},\dots,\phitil_{K}$ are admissible function and achieves the supremum in \cref{eq:duality}, we have for any $y(1),\dots, y(k)$
\begin{align*}
    \sum_{l=1}^{k-1} \phistar_l(y(l)) + \sum_{l=k}^{K} \phitil_l(y(l)) \leq \sum\limits_{i\neq j} a_{ij}y(i)y(j).
\end{align*}
This implies 
\begin{align*}
    2\inner*{y(k), \sum\limits_{l\neq k}a_{kl} y(l) } + \sum\limits_{\substack{i,j:i\neq j,\\ i\neq k,j\neq k}} a_{ij}\inner{y(i), y(j)}  -   \sum_{l=1}^{k-1} \phistar_l(y(l)) - \sum_{l=k+1}^{K} \phitil_l(y(l)) \geq   \phitil_k(y(k)).
\end{align*}
Then take infimum over $y(1),\dots,y(k-1), y(k+1),\dots,y(K)$, we have
\begin{align*}
    \phistar_k(y(k))  \geq \phitil_k(y(k)).
\end{align*}
This implies that $\phistar_1,\dots,\phistar_{k},\phitil_{k+1},\dots,\phitil_{K}$ achieves the supremum in \cref{eq:duality}. Thus by induction, we have shown that $\phistar_1,\dots,\phistar_{K}$ are admissible function and achieves the supremum in \cref{eq:duality}. Furthermore, since for any $k$, $\phi_k$ are infimum of linear functions, the Lipschitz constant for $\phi_k$ is thus upper bounded by 
\begin{align*}
    \sup_{\substack{y(1),\dots,y(k-1),\\y(k+1),\dots,y(K)\\
    \in \supp(\gmstar_A(\mu(1),\dots,\mu(K)))}} 2\norm*{\sum\limits_{l\neq k}a_{kl} y(l)} \leq  2 \eta_{K,d,A}\sqrt{K}.
\end{align*}
Therefore, 
    \begin{align*}
        \revindent[.5]\QMW(\muhat_n(1),\dots \muhat_n(1);A ) - \QMW(\mustar(1),\dots,\mustar(K);A) \\
        &= \sum\limits_{k=1}^{K} \int a_{kk} \norm{y(k)}^2 \prn*{\muhat_n(k,dy(k))  - \mustar(k,dy(k))} \\
        &\quad\quad+   \sup_{\substack{\phi_1,\dots,\phi_K:\\ \sum\limits_{k=1}^{K}\phi_k(y(k)) \leq \sum\limits_{i\neq j} a_{ij}y^\top(i)y(j) \\\phi_k: 2\sqrt{K}\eta_{K,d,A}\text{-Lip.} }} \sum\limits_{k=1}^{K} \int \phi_k(y(k)) \muhat_n(k,dy(k)) \\
        &\quad\quad-  \sup_{\substack{\phi_1,\dots,\phi_K:\\ \sum\limits_{k=1}^{K}\phi_k(y(k)) \leq \sum\limits_{i\neq j} a_{ij}y^\top(i)y(j) \\\phi_k: 2\sqrt{K}\eta_{K,d,A}\text{-Lip.}}} \sum\limits_{k=1}^{K} \int \phi_k(y(k)) \mustar(k,dy(k)) \\
        &\le  \sum\limits_{k=1}^{K} \int a_{kk} \norm{y(k)}^2  \prn*{\muhat_n(k,dy(k))  - \mustar(k,dy(k))} \\
        & \quad\quad+ \sup_{\substack{\phi_1,\dots,\phi_K:\\ \sum\limits_{k=1}^{K}\phi_k(y(k)) \leq \sum\limits_{i\neq j} a_{ij}y^\top(i)y(j) \\\phi_k: 2\sqrt{K}\eta_{K,d,A}\text{-Lip.} }} \sum\limits_{k=1}^{K} \int \phi_k(y(k)) (\muhat_n(k,dy(k)) - \mustar(k,dy(k)))\\
        &\le \abs*{\sum\limits_{k=1}^{K} \int a_{kk} \norm{y(k)}^2  \prn*{\muhat_n(k,dy(k))  - \mustar(k,dy(k))} } \\
        &\quad\quad+ \sum\limits_{k=1}^{K} \sup_{\phi_k\in \Flip[2\sqrt{K}\eta_{K,d,A}]} \int \phi_k(y(k))  \prn*{\muhat_n(k,dy(k))  - \mustar(k,dy(k))}. 
    \end{align*}
    Similarly we can show $-(\QMW(\muhat_n(1),\dots \muhat_n(1);A ) - \QMW(\mustar(1),\dots,\mustar(K);A))$ satisfies the same upper bound.
    By Hoeffding's inequality (Proposition 2.5 of \cite{wainwright2019high}), we have
    \begin{align}
        \En\abs*{\sum\limits_{k=1}^{K} \int a_{kk} \norm{y(k)}^2  \prn*{\muhat_n(k,dy(k))  - \mustar(k,dy(k))} } \lesssim  \prn*{\sum\limits_{k=1}^{K} |a_{kk}|} \cdot  (n/K)^{-1/2} .\label{ineq:lower-order-control}
    \end{align}
    Finally, by \pref{lem:expectation-radermacher}, we obtain
    \begin{align}\label{ineq:lip-part-control}
        \begin{split}
            \revindent[2]\En\brk*{\sum\limits_{k=1}^{K} \sup_{\phi_k\in \Flip[2\sqrt{K}\eta_{K,d,A}]} \int \phi_k(y(k))  \prn*{\muhat_n(k,dy(k))  - \mustar(k,dy(k))}} \\
            &\lesssim K^{3/2}\eta_{K,d,A}\cdot  
            \begin{cases} 
            (n/K)^{-2/d}, & \text{if } d > 4, \\
            (n/K)^{-1/2} \log(n/K), & \text{if } d = 4, \\
            (n/K)^{-1/2}, & \text{if } d < 4,
            \end{cases}.
        \end{split}
    \end{align}
    The bound from \cref{ineq:lower-order-control} is lower order compared to the bound of the mean concentration from \cref{lem:mean-concentration}. 
    In all, combine \pref{lem:mean-concentration}, \cref{ineq:random-control}, and \cref{ineq:lip-part-control}, we have
    \begin{align*}
        \begin{split}
            \revindent[0] \EE\left[\left|\QMW(\muhat_n(1),\dots \muhat_n(K);A ) - \QMW(\mustar(1),\dots,\mustar(K);A)\right|\right] \\
            &\lesssim K\sqrt{\sum\limits_{k=1}^{K}a_{kk}^2}   \cdot  n^{-1/2}  + \prn*{\sum\limits_{k=1}^{K} |a_{kk}|} \wedge
        \begin{cases} 
            \eta_{K,d,A}\cdot K^{3/2} (n/K)^{-2/d}, & \text{if } d > 4, \\
            \eta_{K,d,A}\cdot K^2 n^{-1/2} \log(n/K), & \text{if } d = 4, \\
            \eta_{K,d,A}\cdot K^2n^{-1/2}   , & \text{if } d < 4,
            \end{cases}
        \end{split}
    \end{align*}

\end{proof}

\begin{lemma}
    \label{lem:finite-support-eta-bound}
    For $n,K\geq 2$ and $d\geq 1$, for any $K$ mean zero marginals $\mu(1),\dots,\mu(K)$ finitely supported on $B_1^d(0)\subset \bR^d$, we have
    \begin{align*}
        \sup_{\substack{(y(1),\dots,y(K))\in\\ \supp \prn*{\gmstar(\mu(1),\dots,\mu(K)) }} } \sup_{l\in [K]} \norm*{\sum\limits_{k\neq l}  \frac{y(k)}{\sqrt{K}}  } \leq 4.
    \end{align*}
    $(y(1),\dots,y(K))$ in the support of the optimal coupling $\gmstar(\mu(1),\dots,\mu(K))$ with respect to the objective in \pref{eq:quadratic.objective}.
\end{lemma}

\begin{proof}[\pfref{lem:finite-support-eta-bound}]
Denote the set $\cA_1 = \set*{y: \norm*{\sum\limits_{k=1}^{K} y(k)} \leq 2\sqrt{K} }   $ and $\cA_2 = \set*{y: \norm*{\sum\limits_{k=1}^{K} y(k)} \geq 4 \sqrt{K} -1 }$. 
We prove by contradition. Suppose otherwise, then we have $\gmstar(\cA_2) >0$. On the other hand, we note that the objective value is upper bounded by 
\begin{align*}
    \MW_2^2(\mu(1),\dots,\mu(K)) &= \int \left\| \frac{\sum_{k=1}^{K} y(k)}{K} \right\|_2^2  \gmstar(d y) \\
    &\leq \int \left\| \frac{\sum_{k=1}^{K} y(k)}{K} \right\|_2^2  \prod_{k=1}^K \mu(k, d y(k)),\\
    &=\frac{1}{K^2} \sum\limits_{k=1}^{K}   \int \norm{y(k)}^2  \mu(k, d y(k)) \leq \frac{1}{K}.
\end{align*}
Then by Markov's inequality, we have $\gmstar(\cA_1) > 3/4$. Since 
all the marginals are finitely supported, we can find $y_1\in \cA_1$ and $y_2\in \cA_2$ such that $\gmstar(y_1),\gmstar(y_2)>0$. Let $p = \gmstar(y_1)\wedge\gmstar(y_2)$.
For any $j\in [K]$, we consider $y_{1,j} = (y_1(1),\dots,y_1(j-1), y_2(j),y_1(j+1),\dots,y_1(K))$ and $y_{2,j} = (y_2(1),\dots,y_2(j-1), y_1(j),y_2(j+1),\dots,y_2(K))$. Then since the $K$ couplings $\gmstar - p (\delta_{y_1} + \delta_{y_2})+ p (\delta_{y_{1,j}} + \delta_{y_{2,j}}) $ for $j\in [K]$ are no better than the optimal coupling in the objective value, we thus have for $j\in [K]$,
\begin{align*}
 \norm*{\sum\limits_{k=1}^{K} y_1(k)}^2 +\norm*{\sum\limits_{k=1}^{K} y_2(k)}^2 \leq   \norm*{y_{2}(j)+ \sum_{k\neq j} y_1(k)}^2  +  \norm*{y_{1}(j)+ \sum_{k\neq j} y_2(k)}^2.
\end{align*}
Simplifying the expression, we obtain for any $j\in [K]$,
\begin{align*}
\inner*{ y_1(j) - y_2(j), \sum\limits_{k\neq j} (y_1(k) - y_2(k))  } \leq 0.
\end{align*}
This further implies
\begin{align*}
    \inner*{ y_1(j) - y_2(j), \sum\limits_{k} (y_1(k) - y_2(k))  } \leq \norm{y_1(j) - y_2(j)}^2.
\end{align*}
Summing up over $j\in [K]$, we have
\begin{align*}
\norm*{\sum\limits_{k} (y_1(k) - y_2(k)) }^2 \leq  \sum\limits_{k=1}^{K} \norm*{ (y_1(k) - y_2(k)) }^2 \leq 2K.
\end{align*}
However, we have
\begin{align*}
4\sqrt{K}-1 \leq \norm*{\sum\limits_{k=1}^{K}y_2(k) } \leq \norm*{\sum\limits_{k=1}^{K}y_1(k) } + \norm*{\sum\limits_{k=1}^K (y_1(k) - y_2(k)) } \leq (2+\sqrt{2})\sqrt{K}.
\end{align*}
This is a contradiction. Thus have $\gmstar(\cA_2) = 0$, which means 
\begin{align*}
    \sup_{\substack{(y(1),\dots,y(K))\in\\ \supp \prn*{\gmstar(\mu(1),\dots,\mu(K)) }} } \sup_{k\in [K]} \norm*{\sum\limits_{l\neq k}  \frac{y(l)}{\sqrt{K}}  } \leq 4.
\end{align*}

\end{proof}

\begin{lemma}[Theorem 7.7 of \cite{polyanskiy2022information}; \cite{strassen1965existence}]
    \label{lem:TV-achieving-coupling}
    Suppose $P$ and $Q$ are two distributions on the space $\cX$.
    Provided that the diagonal $\set{(x,x):x\in \cX}$ is measurable. 
    \begin{align*}
    \TV(P,Q) = \min_{P_{X,Y}}\set{  P_{X,Y}[X\neq Y]: P_X = P, P_Y=Q }
    \end{align*}
    where minimization is over the coupling distributions $P_{X,Y}$ with that $P_X = P$ and $P_Y = Q$.
\end{lemma}

\begin{proof}[\pfref{thm:main-upper-bound}]
By applying \cref{prop:quadratic-form-upper-bound} with $A = \mathbbm{1}\mathbbm{1}^\top /K^2$, we first obtain the following convergence rate:
For any $n,d,K$, we have
\begin{align}\label{ineq:MOT.expectation.upper.bound.slow.rate}
\begin{split}
    \revindent[2] \EE\left[\left|\MW_2^2(\muhat_n(1),\dots \muhat_n(K) ) - \MW_2^2(\mustar(1),\dots,\mustar(K))\right|\right] \\
    &\lesssim n^{-1/2} + 
\prn*{1/K} \wedge
\begin{cases} 
    \eta_{K,d}\cdot K^{-1/2} (n/K)^{-2/d}, & \text{if } d > 4, \\
    \eta_{K,d}\cdot n^{-1/2} \log(n/K), & \text{if } d = 4, \\
    \eta_{K,d}\cdot n^{-1/2}  , & \text{if } d < 4,
    \end{cases}
\end{split}
\end{align}
where the notation $\lesssim$ hides constants that only depend on the dimension $d$. 
Moreover, the coefficient $\eta_{K,d}$ is defined as
\begin{align*}
\eta_{K,d} \ldef \sup_{\substack{\mu(1),\dots,\mu(K)\\\text{with zero mean} }} \hspace{3pt} \sup_{\substack{(y(1),\dots,y(K))\in\\ \supp \prn*{\gmstar(\mu(1),\dots,\mu(K)) }} } \sup_{l\in [K]} \norm*{\sum\limits_{k\neq l}  \frac{y(k)}{\sqrt{K}}  },
\end{align*}
where $\mu(1),\dots,\mu(K)$ are $d$-dimensional distributions with zero mean and $(y(1),\dots,y(K))$ in the support of the optimal coupling $\gmstar(\mu(1),\dots,\mu(K))$ with respect to the objective in \pref{eq:quadratic.objective}.
Moreover, since the support is bounded, by \cref{lem:finite-support-eta-bound}, we have, $\eta_{K,d} \leq \frac{K-1}{\sqrt{K}} \leq \sqrt{K}$. Let $N\geq n$ be an integer to be specified later. Let $ \muhat_N= (\muhat_N(1),\dots,\muhat_N(K))$ be the emprical distribution sampled from $\mustar$ where each margin has $N/K$ samples.
Then \cref{ineq:MOT.expectation.upper.bound.slow.rate} implies 
\begin{align*}
    \begin{split}
        \revindent[2] \EE\left[\left|\MW_2^2(\muhat_N(1),\dots \muhat_N(K) ) - \MW_2^2(\mustar(1),\dots,\mustar(K))\right|\right] \\
        &\lesssim N^{-1/2} + 
    \prn*{1/K} \wedge
    \begin{cases} 
        (N/K)^{-2/d}, & \text{if } d > 4, \\
        (N/K)^{-1/2} \log(N/K), & \text{if } d = 4, \\
        (N/K)^{-1/2}  , & \text{if } d < 4.
        \end{cases}
    \end{split}
\end{align*}
Let $\mubar_n= (\mubar_n(1),\dots,\mubar_n(K))$. The distribution $\mubar_n$ is the empirical distribution sampled from $\muhat_N$ with replacement where each margin has $n/K$ samples. Since $\muhat_N$ and $\mubar_n$ are both discretely supported, then rerun the proof of \cref{prop:quadratic-form-upper-bound} with \cref{lem:finite-support-eta-bound} gives the following convergence rate due to improved Lipschitz constant
\begin{align*}
\begin{split}
    \revindent[2] \EE_{\muhat_N, \mubar_n}\left[\left|\MW_2^2(\mubar_n(1),\dots,\mubar_n(K)) - \MW_2^2(\muhat_N(1),\dots \muhat_N(K) ) \right|\right]\\
    &\lesssim n^{-1/2} + 
\prn*{1/K} \wedge
\begin{cases} 
 K^{-1/2} (n/K)^{-2/d}, & \text{if } d > 4, \\
 n^{-1/2} \log(n/K), & \text{if } d = 4, \\
 n^{-1/2}  , & \text{if } d < 4,
    \end{cases}
\end{split}
\end{align*}
 Finally, we consider the variables  $\mutil_n= (\mutil_n(1),\dots,\mutil_n(K))$, which is the emprical distribution sampled from $\muhat_N$ without replacement where each margin has $n/K$ samples. 
We couple $\mutil_n$ and $\mubar_n$ conditioned on $\muhat_N$ by the $\TV$ distance achieving coupling $P_{\mutil_n, \mubar_n}$ as defined in \cref{lem:TV-achieving-coupling}. 
Thus in all, we have
\begin{align*}
    \revindent[.2]\EE_{\muhat_n}\left[\left|\MW_2^2(\muhat_n(1),\dots \muhat_n(K) ) - \MW_2^2(\mustar(1),\dots,\mustar(K))\right|\right] \\
    &=    \EE_{\mutil_n}\left[\left|\MW_2^2(\mutil_n(1),\dots \mutil_n(K) ) - \MW_2^2(\mustar(1),\dots,\mustar(K))\right|\right] \\
    &\leq \EE_{\muhat_N, \mutil_n, \mubar_n}\left[\left|\MW_2^2(\muhat_N(1),\dots \muhat_N(K) ) - \MW_2^2(\mustar(1),\dots,\mustar(K))\right|\right] \\
    &\quad + \EE_{\muhat_N, \mutil_n, \mubar_n}\left[\left|\MW_2^2(\muhat_N(1),\dots \muhat_N(K) ) - \MW_2^2(\mubar_n(1),\dots,\mubar_n(K))\right|\right] \\
    &\quad + \EE_{\muhat_N, \mutil_n, \mubar_n}\left[\left|\MW_2^2(\mubar_n(1),\dots,\mubar_n(K)) - \MW_2^2(\mutil_n(1),\dots \mutil_n(K) ) \right|\right] \\
    &\lesssim  \TV(\mubar_n, \mutil_n) + N^{-1/2} + n^{-1/2} + 
    \begin{cases} 
        (N/K)^{-2/d} + K^{-1/2} (n/K)^{-2/d}, & \text{if } d > 4, \\
        (N/K)^{-1/2} \log(N/K) + n^{-1/2} \log(n/K), & \text{if } d = 4, \\
        (N/K)^{-1/2} +n^{-1/2} , & \text{if } d < 4,
    \end{cases}
\end{align*}
where the first equality is due to the fact that $\muhat_n$ and $\mutil_n$ is equal in distribution, the second inequality is by triangle inequality and the third inequality is by combing the aforementioned rates. Now we finally show $\TV(\mubar_n, \mutil_n)$ is $o(1)$ as $N\to \infty$ and then take $N\to \infty$ to obtain the desired result.
Concretely, let event $\cE_{N,n,K} = \set{\text{no points in $\muhat_N$ are sampled twice when sampling $\mubar_n$ with replacement}}$. Then we have $\mubar_n|\cE_{N,n,K} = \mutil_n$. Thus,
\begin{align*}
    \TV(\mubar_n, \mutil_n) \leq  1-\bP_{\mubar_n}(\cE_{N,n,K})  = o(1), \quad \text{as }N\to \infty.
\end{align*}
This concludes our proof.

\end{proof}

\begin{proof}[\pfref{prop:quadratic-form-upper-bound-simplified}]
    This is a direct corollary of \cref{prop:quadratic-form-upper-bound}, by the upper bound of 
    \begin{align*}
        \eta_{K,d,A} \lesssim \norm*{A}_{1,\infty}/\sqrt{K}.
    \end{align*}
\end{proof}

\subsection{Lower bound}
\label{app:lower-bound}

\begin{proof}[\pfref{prop:minimax-lower-bound}]
We only need to prove for $d=1$. 
For any $K\geq 2$, if $K$ is even, then
consider $\mu(1)= \dots = \mu(K) = Ber(1/2) $ and $\mu'(1) = \dots = \mu'(K) = Ber(1/2+\veps)$ with some $\veps>0$ to be decided later. If $K$ is odd, then set the $\mu(K) = \mu'(K)$ to be the point mass distribution on $0$ and we reduce back to the even case. 
We first show that $\MW_2^2(\mu(1),\dots, \mu(K)) = 1/4$ and $\MW_2^2(\mu'(1),\dots, \mu'(K)) \geq 1/4 + \veps$. Indeed, by \Cref{lem:mean-decomposition}, we have
\begin{align*}
    \MW_2^2(\mu(1),\dots, \mu(K)) \geq  \prn*{ \frac{\sum_{k=1}^{K} \EE_{\mu(k)} [X]}{K}}^2 = 1/4
\end{align*} 
and 
\begin{align*}
    \MW_2^2(\mu'(1),\dots, \mu'(K)) \geq \prn*{ \frac{\sum_{k=1}^{K} \EE_{\mu(k)} [X]}{K} }^2 \geq  1/4 + \veps.
\end{align*}
Meanwhile, for $\mu(1),\dots, \mu(K)$ the coupling $\gamma$ that assigns $Y(1) = Y(3) = \dots = Y(K-1)$ and $Y(2) = Y(4) = \dots = Y(K) = 1- Y(1)$ is an admissible coupling that achieves
\begin{align*}
\int \prn*{ \sum\limits_{k=1}^{K} y(K)/K }^2  \gamma(dy) = 1/4.
\end{align*}
Thus we have shown $\MW_2^2(\mu(1),\dots, \mu(K)) = 1/4$ and $\MW_2^2(\mu'(1),\dots, \mu'(K)) \geq 1/4 + \veps$.
Let $\bP$ be the distribution on the observations $Y_i(1),\dots,Y_i(K)$ with $\mustar(1),\dots, \mustar(K) = \mu(1),\dots, \mu(K)$ and $\bP'$ be the distribution on the observations $Y_i(1),\dots,Y_i(K)$ with $\mustar(1),\dots, \mustar(K) = \mu'(1),\dots, \mu'(K)$. Then we have
\begin{align*}
\TV(\bP,\bP') \lesssim \sqrt{ \KL (\bP,\bP') } \lesssim \sqrt{n \cdot\KL (Ber(1/2),Ber(1/2+\veps)) } \lesssim \sqrt{n\veps^2}.
\end{align*}
Thus by choosing $\veps = c\cdot n^{-1/2})$ for constant $c$ small enough, we can make $\TV(\bP,\bP') \leq 1/2$ while $\MW_2^2(\mu'(1),\dots, \mu'(K)) - \MW_2^2(\mu(1),\dots, \mu(K)) \gtrsim \Omega(n^{-1/2})$. Then we apply Le Cam's two-point method (Lemma 1 of \cite{yu1997assouad}) and obtain a lower bound of $n^{-1/2}$. Furthermore, since our problem has $K$ margins, it is easy to embed the 1-Wasserstein distance estimation problem in the first two margins. Thus, by the lower bound of Wasserstein distance estimation from theorem 22 of \cite{manole2024sharp}, we obtain a lower bound of $K^{-2}(n/K)^{-2/d}$ when $d>4$.
\end{proof}

\section{Algorithms from Section \ref{sec:empirical}} \label{app:algorithm}

In this section, we introduce the MOT Greenkhorn algorithm and MOT Sinkhorn algorithm that we implemented. We then specify the following three properties we use regarding the MOT Sinkhorn algorithm: (1) The convergence guarantee. (2) The application of log-exp-sum trick. (3) The lower bounds without convergence.

For any $\gamma$ a $K$-marginal tensor, define 
\begin{align*}
\MTV(\gamma, \mu(1),\dots,\mu(K)) \ldef \sum\limits_{k=1}^{K} \int_{y(k)} |\gamma(k,dy(k))- \mu(k,dy(k))|,
\end{align*}
where $\gamma(k, dy(k)) = \int_{y(1),...,y(k-1),y(k+1),...,y(K)} \gamma(dy(1),...,dy(K) ) $ which might be a signed measure in general.

\paragraph{Multi-marginal Greenkhorn}
We follow \cite{altschuler2017near} and introduce the algorithm of MOT Greenkhorn. The function $\rho(a,b)$ in \pref{line:greedy} of \Cref{alg:greenkhorn} is defined as $\rho(a,b) = b-a+a\log (a/b)$.

\begin{algorithm}[H]
\caption{Multi-marginal Greenkhorn}
\label{alg:greenkhorn}
\begin{algorithmic}[1]
\Require Cost matrix $C$, Multiple marignals $\mu(1)$, $\mu(2)$, \dots,$\mu(K)$, Accuracy $\veps>0$. 
\State $\eta \lto  4 \sum\nolimits_{k=1}^{K}\log n_k/\veps$, $\veps'\lto \frac{\veps}{8 \norm{C}_\infty}$, $A\lto \exp(-\eta C)$.
\State $m(k) \lto \vz\in \bR^{n_k}$ for $k\in [K]$, and $\gamma \lto A/\norm{A}_1$. \label{line:init}
\While{$\MTV(\gamma, \mu(1),...,\mu(K) )) > \veps'$} \label{line:while-loop-condition}
\smskip
\State $(k^\star, i^\star) \lto \argmax_{k\in [K],i\in[n_k]} \rho(\mu(k,i), \gamma(k,i))$. \label{line:greedy}
 \Comment{$\mu (k,i)$ is $1/n_k$ in our case.}
\smskip
\State $m(k^\star, i^\star)  \lto m(k^\star,i^\star) + \log \mu(k^\star,i^\star) - \log \gamma (k^\star,i^\star) $. \label{line:Greenkhorn-m}
\smskip
\State \multiline{$\gamma(i_1,\dots,i_K)= A(i_1,\dots,i_K) \cdot \prod_{k=1}^K \exp(m(k,i_k)) $ for all $i_1,\dots,i_K\in [n_1]\times\cdots\times[n_K]$.} 
\EndWhile \label{line:endwhile}
\State Let $\gmhat^{(0)} = \gamma$.
\For{$k=1,...,K$}
\State $v(k) \lto (\mu(k)/\gmhat^{(k-1)}(k)) \wedge \mathbbm{1}$, where the division and $\wedge$ are elementwise operation.
\State $\gmhat^{(k)}(i_1,\dots,i_K) =  \gmhat^{(k-1)}(i_1,\dots,i_K) \cdot v(k,i_k)$ for all $i_1,\dots,i_K\in [n_1]\times\cdots\times[n_K]$.
\EndFor
\For{$k=1,...,K$}
\State $\err_k \lto \mu(k) -  \gmhat^{(K)}(k) $.
\EndFor
\State \Return $\gmhat(i_1,...,i_K) \lto \gmhat^{(K)}(i_1,...,i_K) +  \prod_{k=1}^K\err_k(i_k)/ \norm{\err_K}_1^{K-1}$ for all $i_1,...,i_K$.   
\end{algorithmic}
\end{algorithm}

\paragraph{Multi-marginal Sinkhorn} We use the Multi-marginal Sinkhorn algorithm as in \Cref{alg:main-alg}
.
\algnewcommand{\IfThen}[2]{%
  \State \algorithmicif\ #1\ \algorithmicthen\ #2}

\begin{algorithm}[!htp]
\caption{Multi-marginal Sinkhorn}
\label{alg:main-alg}
\begin{algorithmic}[1]
    \setstretch{1.1}
\State \textbf{Parameters}: Cost tensor $C\in \bR^{n_1\times \dots \times n_k}$, Marignals $\mu(1),\dots,\mu(K)$, Targeted accuracy $\veps>0$. 
\State Let $\eta = \prn{4 \sum_{k=1}^{K}\log n_k}/\veps$, $\veps' = \veps/(8\norm{C}_\infty)$, and $A = \exp(-\eta C)$ .
\State Let $t=0$, $m^{(0)}(k) = \vz \in\bR^{n_k}$ for $k\in [K]$, and $\gamma^{(0)}= A/\norm{A}_1$. 
\For{$t=1,2,\dots$}
\State Let $k_t = \argmax_k \sum_{i_k\in [n_k]} \mu(k,i_k)\log ( \mu(k,i_k)/\gamma^{(t-1)}(k,i_k) )  $.
\State $m^{(t)}(k_t) = m^{(t-1)}(k_t) + \log \mu(k_t) - \log \gamma^{(t-1)}(k_t)$. \label{line:updating-m}
\State \multiline{$\gamma^{(t)}(i_1,\dots,i_K)= A(i_1,\dots,i_K) \cdot \prod_{l=1}^k \exp(m^{(t)}(l,i_l)) \prod_{l=k+1}^{K} \exp(m^{(t-1)}(l,i_l)) $ for all $i_1,\dots,i_K\in [n_1]\times\cdots\times[n_K]$.} \label{line:log-sum-exp}
\IfThen{ $\MTV(\gamma^{(t)}, \mu(1),\dots,\mu(K)) \leq \veps'$ }{\textbf{break}}
\EndFor
\State Let $\gmhat^{(0)} = \gamma^{(t)}$.
\For{$k=1,...,K$}
\State $v(k) \lto (\mu(k)/\gmhat^{(k-1)}(k)) \wedge \mathbbm{1}$, where the division and $\wedge$ are elementwise operation.
\State $\gmhat^{(k)}(i_1,\dots,i_K) =  \gmhat^{(k-1)}(i_1,\dots,i_K) \cdot v(k,i_k)$ for all $i_1,\dots,i_K\in [n_1]\times\cdots\times[n_K]$.
\EndFor
\For{$k=1,...,K$}
\State $\err_k \lto \mu(k) -  \gmhat^{(K)}(k) $.
\EndFor
\State \Return $\gmhat(i_1,...,i_K) \lto \gmhat^{(K)}(i_1,...,i_K) +  \prod_{k=1}^K\err_k(i_k)/ \norm{\err_K}_1^{K-1}$ for all $i_1,...,i_K$.   
\end{algorithmic}
\end{algorithm}

The guarantees that are interesting for us are shown in the following:

\subsection{Convergence}
\begin{theorem}[Theorem 16 of \cite{lin2022complexity}]
    When \Cref{alg:main-alg} converges, we have
    \begin{align*}
        \inf_{\gamma\in \Gamma(\mu(1),...,\mu(K)) } \int_y C(y) \gamma(d y) \geq \int_y C(y) \gmhat(d y)  -\veps.
    \end{align*}
\end{theorem}

\subsection{Application of log-sum-exp trick}

\paragraph{Log-sum-exp trick} Let $f:\cX\to \bR$ be any function with $\cX$ being any finite set.  
The celebrated log-sum-exp trick is used for computing quantities of the type $\log \sum_{x\in \cX} \exp (f(x))$. It computes the equivalent form of $\max_x f(x) +  \log\sum_{x} \exp (f(x) - \max_{x'} f(x'))$, which is much more stable computationally since $\sum_{x} \exp (f(x) - \max_{x'} f(x'))\geq 1$ and $f(x) - \max_{x'} f(x') \leq 0$ for all $x\in \cX$.

In the implementation of \Cref{alg:main-alg}, the log-sum-exp trick can be apply to compute the update of $m^{(t)}(k)$:
For \Cref{alg:main-alg} specifically, the computation in \pref{line:log-sum-exp} is unstable when $\veps$ is chosen small because it involves exponential to the power of magnitude $\Omega(\eta) = \Omega(1/\veps)$. 
We can apply the celebrated log-sum-exp trick to stabilize the algorithm by skipping \pref{line:log-sum-exp} and updating $m^{(t)}(k_{t})$ in \pref{line:updating-m} for the next timestep by
\begin{align*}
m^{(t)}(k_{t}) &= m^{(t-1)}(k_{t}) + \log \mu(k_{t}) &\\
&- \log \sum\limits_{\substack{i_1,\dots,i_{k_{t}-1}\\i_{k_{t}+1},\dots,i_K}} \exp \prn*{ -\eta C(i_1,\dots,i_{k_{t}-1},\cdot, i_{k_{t}+1},\dots,i_K)  + \sum\limits_{l\neq k_t} m^{(t-1)}(l,i_l)}, 
\end{align*}
where the final term is executed using the log-sum-exp trick.

\subsection{Lower bounds without convergence}

\begin{lemma}[Section 4.4 of \cite{peyre2019computational}]
    For any $t\geq 1$, \Cref{alg:main-alg} satisfies
    \begin{align*}
        \inf_{\gamma\in \Gamma(\mu(1),...,\mu(K)) } \int_y C(y) \gamma(d y) \geq \frac{1}{\eta} \sum_{k=1}^K \int_{y(k)} m^{(t)}(y(k)) \mu(k,dy(k)).
    \end{align*}
\end{lemma}

\section{Description of datasets}\label{sec:datasets}

\subsection{Epitaxial layer growth data}
In fabricating integrated circuit (IC) devices, an epitaxial layer is grown on polished silicon wafers mounted on a six-faceted susceptor within a metal bell jar. Chemical vapors are injected and heated until the layer reaches a desired thickness. This dataset originates from a factorial design experiment investigating the impact of various factors to the thickness (outcome). The experiment data we use considers two treatments: susceptor rotation method and nozzle position, resulting in a maximum of four margins.

\subsection{Education data}
The Student Achievement and Retention Project (Project STAR) is a randomized evaluation of academic services and incentives at one of the satellite campuses of a large Canadian university. The STAR demonstration randomly assigned entering first-year undergraduates to one of four arms: a service strategy known as the Student Support Program (SSP), an incentive strategy known as the Student Fellowship Program (SFP), an intervention offering both known as the SFSP, and a control strategy offering neither of the two programs. As a result, there are four margins in total. Several outcomes are of interest and analyzed: the GPA at the end of year 1, the GPA at the end of year 2, whether the student is in good standing at the end of year 1,whether the student is in good standing at the end of year 2, grade of the first semester in the first year.

\subsection{Helpfulness data}
The study is a randomized experiment to learn subjects with uniform experience are more susceptible to self-serving motivated reasoning in both their empathy beliefs and redistribution choices. This experiment employs a 2-by-2 design with two treatments: experience variation—subjects either have varied experience or uniform experience; timing of belief elicitation—time to tell the participant that there is a chance to take on a partner’s workload (corresponding to wealth redistribution). As a result, there are four margins. The outcome is the willingness of sharing the workload/redistributing the wealth.

\end{document}

%% file: arxiv_style.tex
\usepackage[letterpaper, left=1in, right=1in, top=1in, bottom=1in]{geometry}

\PassOptionsToPackage{hypertexnames=false}{hyperref}  %
\usepackage{parskip}

\usepackage[dvipsnames]{xcolor}
\usepackage[colorlinks=true, linkcolor=blue!70!black, citecolor=blue!70!black,urlcolor=black,breaklinks=true]{hyperref}
\usepackage{microtype}
\usepackage{hhline}

\makeatletter
\newcommand{\neutralize}[1]{\expandafter\let\csname c@#1\endcsname\count@}
\makeatother

\usepackage{algorithm}

\usepackage{amsthm}
\usepackage{mathtools}
\usepackage{amsmath}
\usepackage{bbm}
\usepackage{amsfonts}
\usepackage{amssymb}

\usepackage{xpatch}

\usepackage{thmtools}
\usepackage{thm-restate}
\declaretheorem[name=Theorem,parent=section]{theorem}
\declaretheorem[name=Lemma,parent=section]{lemma}
\declaretheorem[name=Assumption, parent=section]{assumption}
\declaretheorem[name=Condition, parent=section]{condition}
\declaretheorem[qed=$\triangleleft$,name=Example,style=definition, parent=section]{example}

\declaretheorem[name=Proposition, parent=section]{proposition}

\makeatletter
  \renewenvironment{proof}[1][Proof]%
  {%
   \par\noindent{\bfseries\upshape {#1.}\ }%
  }%
  {\qed\newline}
  \makeatother

\theoremstyle{definition}  %

\newtheorem{corollary}{Corollary}[section]

\theoremstyle{plain}
\newtheorem{definition}{Definition}[section]

\xpatchcmd{\proof}{\itshape}{\normalfont\proofnameformat}{}{}
\newcommand{\proofnameformat}{\bfseries}

\usepackage{xparse}

\ExplSyntaxOn
\DeclareDocumentCommand{\XDeclarePairedDelimiter}{mm}
 {
  \__egreg_delimiter_clear_keys: %
  \keys_set:nn { egreg/delimiters } { #2 }
  \use:x %
   {
    \exp_not:n {\NewDocumentCommand{#1}{sO{}m} }
     {
      \exp_not:n { \IfBooleanTF{##1} }
       {
        \exp_not:N \egreg_paired_delimiter_expand:nnnn
         { \exp_not:V \l_egreg_delimiter_left_tl }
         { \exp_not:V \l_egreg_delimiter_right_tl }
         { \exp_not:n { ##3 } }
         { \exp_not:V \l_egreg_delimiter_subscript_tl }
       }
       {
        \exp_not:N \egreg_paired_delimiter_fixed:nnnnn 
         { \exp_not:n { ##2 } }
         { \exp_not:V \l_egreg_delimiter_left_tl }
         { \exp_not:V \l_egreg_delimiter_right_tl }
         { \exp_not:n { ##3 } }
         { \exp_not:V \l_egreg_delimiter_subscript_tl }
       }
     }
   }
 }

\keys_define:nn { egreg/delimiters }
 {
  left      .tl_set:N = \l_egreg_delimiter_left_tl,
  right     .tl_set:N = \l_egreg_delimiter_right_tl,
  subscript .tl_set:N = \l_egreg_delimiter_subscript_tl,
 }

\cs_new_protected:Npn \__egreg_delimiter_clear_keys:
 {
  \keys_set:nn { egreg/delimiters } { left=.,right=.,subscript={} }
 }

\cs_new_protected:Npn \egreg_paired_delimiter_expand:nnnn #1 #2 #3 #4
 {%
  \mathopen{}
  \mathclose\c_group_begin_token
   \left#1
   #3
   \group_insert_after:N \c_group_end_token
   \right#2
   \tl_if_empty:nF {#4} { \c_math_subscript_token {#4} }
 }
\cs_new_protected:Npn \egreg_paired_delimiter_fixed:nnnnn #1 #2 #3 #4 #5
 {
  \mathopen{#1#2}#4\mathclose{#1#3}
  \tl_if_empty:nF {#5} { \c_math_subscript_token {#5} }
 }
\ExplSyntaxOff

\XDeclarePairedDelimiter{\supnorm}{
  left=\lVert,
  right=\rVert,
  subscript=\infty
  }

%% file: jian.tex
\usepackage{comment}

\def\smskip{\vskip 5 pt}
\def\medskip{\vskip 10 pt}

\DeclarePairedDelimiter{\abs}{\lvert}{\rvert} %
\DeclarePairedDelimiter{\brk}{[}{]}

\DeclarePairedDelimiter{\prn}{(}{)}
\DeclarePairedDelimiter{\norm}{\|}{\|}
\DeclarePairedDelimiter{\inner}{\langle}{\rangle}
\DeclarePairedDelimiter{\set}{\{}{\}}

\newcommand{\muhat}{\hat{\mu}}
\newcommand{\En}{\mathbb{E}}
\newcommand{\cR}{\mathcal{R}}
\renewcommand{\d}{\textnormal{d}}
\newcommand{\ldef}{\vcentcolon=}
\newcommand{\eps}{\epsilon}

\newcommand{\unif}{\mathrm{Unif}}
\newcommand{\gammahat}{\hat{\gamma}}

\newcommand{\revindent}[1][1]{\hspace{#1in}&\hspace{-#1in}}

\usepackage{prettyref}
\newcommand{\pref}[1]{\prettyref{#1}}
\newcommand{\pfref}[1]{Proof of \Cref{#1}}
\newcommand{\savehyperref}[2]{\texorpdfstring{\hyperref[#1]{#2}}{#2}}
\newrefformat{eq}{\savehyperref{#1}{Eq.~\textup{(\ref*{#1})}}}
\newrefformat{eqn}{\savehyperref{#1}{Equation~\ref*{#1}}}
\newrefformat{lem}{\savehyperref{#1}{Lemma~\ref*{#1}}}
\newrefformat{lemma}{\savehyperref{#1}{Lemma~\ref*{#1}}}
\newrefformat{def}{\savehyperref{#1}{Definition~\ref*{#1}}}
\newrefformat{line}{\savehyperref{#1}{Line~\ref*{#1}}}
\newrefformat{thm}{\savehyperref{#1}{Theorem~\ref*{#1}}}
\newrefformat{corr}{\savehyperref{#1}{Corollary~\ref*{#1}}}
\newrefformat{cor}{\savehyperref{#1}{Corollary~\ref*{#1}}}
\newrefformat{sec}{\savehyperref{#1}{Section~\ref*{#1}}}
\newrefformat{subsec}{\savehyperref{#1}{Section~\ref*{#1}}}
\newrefformat{app}{\savehyperref{#1}{Appendix~\ref*{#1}}}
\newrefformat{assum}{\savehyperref{#1}{Assumption~\ref*{#1}}}
\newrefformat{ex}{\savehyperref{#1}{Example~\ref*{#1}}}
\newrefformat{fig}{\savehyperref{#1}{Figure~\ref*{#1}}}
\newrefformat{alg}{\savehyperref{#1}{Algorithm~\ref*{#1}}}
\newrefformat{rem}{\savehyperref{#1}{Remark~\ref*{#1}}}
\newrefformat{conj}{\savehyperref{#1}{Conjecture~\ref*{#1}}}
\newrefformat{prop}{\savehyperref{#1}{Proposition~\ref*{#1}}}
\newrefformat{proto}{\savehyperref{#1}{Protocol~\ref*{#1}}}
\newrefformat{prob}{\savehyperref{#1}{Problem~\ref*{#1}}}
\newrefformat{claim}{\savehyperref{#1}{Claim~\ref*{#1}}}
\newrefformat{que}{\savehyperref{#1}{Question~\ref*{#1}}}
\newrefformat{op}{\savehyperref{#1}{Open Problem~\ref*{#1}}}
\newrefformat{fn}{\savehyperref{#1}{Footnote~\ref*{#1}}}
\newrefformat{tab}{\savehyperref{#1}{Table~\ref*{#1}}}
\newrefformat{fig}{\savehyperref{#1}{Figure~\ref*{#1}}}
\newrefformat{proc}{\savehyperref{#1}{Procedure~\ref*{#1}}}
\newrefformat{property}{\savehyperref{#1}{Property~\ref*{#1}}}

\usepackage{xargs}

\newcommand{\mustar}{\mu_\star}
\newcommand{\QMW}{\mathrm{QMW}}
\newcommand{\MW}{\mathrm{MW}}
\newcommand{\gmstar}{\gamma^\star}
\newcommand{\supp}{\mathrm{supp}}
\newcommand{\Flip}[1][L]{\cF_{\mathrm{Lip},#1}}

\newcommand{\mubar}{\bar{\mu}}
\newcommand{\phistar}{\phi^\star}
\newcommand{\phitil}{\wt{\phi}}

\newcommand{\veps}{\varepsilon}
\newcommand{\lto}{\leftarrow}

\newcommand{\vz}{\mathbf{0}}

\newcommand{\err}{\mathrm{err}}

\newcommand{\KL}{\mathrm{KL}}

\newcommand{\MTV}{\mathrm{MTV}}
\newcommand{\gmhat}{\gammahat}

\newcommand{\mutil}{\wt{\mu}}

\newcommand{\multiline}[1]{\parbox[t]{\dimexpr\linewidth-\algorithmicindent}{#1}}

\newcommand{\cX}{\mathcal{X}}
\newcommand{\cA}{\mathcal{A}}
\newcommand{\cE}{\mathcal{E}}